\documentstyle[12pt,fleqn]{article}

\font\egtgot =eufm8 \font\sevgot =eufm7

\font\twlmsb =msbm10 scaled \magstep1 \font\egtmsb =msbm8
\font\sevmsb =msbm7

\newfam\gotfam
\def\pgot{\fam\gotfam\twlgot}
\textfont\gotfam\twlgot \scriptfont\gotfam\egtgot
\scriptscriptfont\gotfam\sevgot
\def\got{\protect\pgot}

\newfam\msbfam
\textfont\msbfam\twlmsb \scriptfont\msbfam\egtmsb
\scriptscriptfont\msbfam\sevmsb
\def\Bbb{\protect\pBbb}
\def\pBbb{\relax\ifmmode\expandafter\Bb\else\typeout{You cann't use
Bbb in text mode}\fi}
\def\Bb #1{{\fam\msbfam\relax#1}}

\def\op#1{\mathop{{\it\fam0} #1}\limits}

\newcommand{\di}{{\rm dim\,}}

\newcommand{\Ker}{{\rm Ker\,}}
\newcommand{\im}{{\rm Im\, }}
\newcommand{\hm}{{\rm Hom\,}}

\newcommand{\nm}[1]{\mid {#1}\mid}

\newcommand{\bll}{\bullet}
\newcommand{\beq}{\begin{equation}}
\newcommand{\eeq}{\end{equation}}
\newcommand{\ben}{\begin{eqnarray}}
\newcommand{\een}{\end{eqnarray}}
\newcommand{\be}{\begin{eqnarray*}}
\newcommand{\ee}{\end{eqnarray*}}
\newcommand{\bea}{\begin{eqalph}}
\newcommand{\eea}{\end{eqalph}}

\newcommand{\lto}{{\leftarrow}}

\newcommand{\gO}{{\got G}}
\newcommand{\cO}{{\cal O}}
\newcommand{\cA}{{\cal A}}

\newcommand{\cG}{{\got g}}
\newcommand{\gd}{{\got d}}

\newcommand{\gS}{{\got S}}

\newcommand{\gQ}{{\got Q}}

\newcommand{\gA}{{\got A}}

\newcommand{\nw}[1]{[{#1}]}

\newcommand{\cP}{{\cal P}}
\newcommand{\cR}{{\cal R}}
\newcommand{\cL}{{\cal L}}
\newcommand{\cV}{{\cal V}}

\newcommand{\cQ}{{\cal Q}}
\newcommand{\cE}{{\cal E}}

\newcommand{\cC}{{\cal C}}

\newcommand{\cK}{{\cal K}}

\newcommand{\bu}{{\bf u}}

\newcommand{\cS}{{\cal S}}

\newcommand{\bL}{{\bf L}}

\newcommand{\al}{\alpha}

\newcommand{\dl}{\delta}
\newcommand{\la}{\lambda}
\newcommand{\La}{\Lambda}
\newcommand{\f}{\phi}

\newcommand{\om}{\omega}

\newcommand{\m}{\mu}

\newcommand{\G}{\Gamma}
\newcommand{\e}{\epsilon}
\newcommand{\ve}{\varepsilon}
\newcommand{\th}{\theta}

\newcommand{\vr}{\varrho}
\newcommand{\up}{\upsilon}
\newcommand{\vt}{\vartheta}

\newcommand{\si}{\sigma}
\newcommand{\Si}{\Sigma}

\newcommand{\bb}{{\bf 1}}

\newcommand{\w}{\wedge}

\newcommand{\wt}{\widetilde}
\newcommand{\wh}{\widehat}
\newcommand{\ol}{\overline}

\newcommand{\dr}{\partial}

\newcommand{\rdr}{\stackrel{\leftarrow}{\dr}{}}
\newcommand{\llr}{\op\longleftarrow}

\newcommand{\ar}{\op\longrightarrow}

\newcommand{\ot}{\otimes}
\newcommand{\ap}{\approx}

\newenvironment{eqalph}{\stepcounter{equation}
\setcounter{equationa}{\value{equation}} \setcounter{equation}{0}

\begin{eqnarray}}{\end{eqnarray}\setcounter{equation}{\value{equationa}}}

\newcounter{example}
\newcounter{remark}
\newcounter{theorem}
\newcounter{proposition}
\newcounter{lemma}
\newcounter{corollary}
\newcounter{definition}
\newcounter{note}

\def\theremark{\arabic{remark}}

\def\thetheorem{\arabic{theorem}}

\def\thedefinition{\arabic{theorem}}

\newenvironment{proof}{{\bf Proof.}}{
\hfill $\Box$ }
\newenvironment{rem}{\refstepcounter{remark} \medskip {\bf Remark
\theremark.} }{ }

\newenvironment{theo}{\refstepcounter{theorem} \medskip{\bf
Theorem \thetheorem.}\it}{ }
\newenvironment{prop}{\refstepcounter{theorem} \medskip{\bf
Proposition \thetheorem.}\it}{ }
\newenvironment{note}{\refstepcounter{theorem} \medskip{\bf
Condition \thetheorem.}\it}{}
\newenvironment{lem}{\refstepcounter{theorem} \medskip{\bf Lemma
\thetheorem.}\it }{}
\newenvironment{cor}{\refstepcounter{theorem} \medskip{\bf
Corollary \thetheorem.} \it}{}

\newcommand{\mar}[1]{}

\begin{document}
\hbox{}

{\parindent=0pt

{\large \bf Grassmann-graded Lagrangian theory of even and odd variables}
\bigskip

{\it G. SARDANASHVILY}

\medskip

{\it Department of Theoretical Physics, Moscow State University,
117234 Moscow, Russia}

\bigskip

\bigskip

\begin{small}

{\bf Abstract.} Graded Lagrangian formalism in terms of a
Grassmann-graded variational bicomplex on graded manifolds is
developed in a very general setting. This formalism provides the
comprehensive description of reducible degenerate Lagrangian
systems, characterized by hierarchies of non-trivial higher-order
Noether identities and gauge symmetries. This is a general case of
classical field theory and Lagrangian non-relativistic mechanics.

\end{small}

}

\section{Introduction}

Conventional Lagrangian formalism on fibre bundles $Y\to X$ over a
smooth manifold $X$ is formulated in algebraic terms of a
variational bicomplex of exterior forms on jet manifolds of
sections of $Y\to X$
\cite{ander,bau,jmp,cmp04,book09,olv,tul,tak2}. The cohomology of
this bicomplex provides the global first variational formula for
Lagrangians and Euler--Lagrange operators, without appealing to
the calculus of variations. For instance, this is the case of
classical field theory if $\di X>1$ and non-autonomous mechanics
if $X=\Bbb R$ \cite{book09,book10,sard08}.

However, this formalism is not sufficient in order to describe
reducible degenerate Lagrangian systems whose degeneracy is
characterized by a hierarchy of higher order Noether identities.
They constitute the Kozul--Tate chain complex whose cycles are
Grassmann-graded elements of certain graded manifolds
\cite{lmp08,jmp05a,book09}. Moreover, many field models also deal
with Grassmann-graded fields, e.g., fermion fields, antifields and
ghosts \cite{book09,gom,sard08}.

These facts motivate us to develop graded Lagrangian formalism of
even and odd variables \cite{lmp08,cmp04,book09,ijgmmp07}.

Different geometric models of odd variables are described either
on graded manifolds or supermanifolds. Both graded manifolds and
supermanifolds are phrased in terms of sheaves of graded
commutative algebras \cite{bart,book09}. However, graded manifolds
are characterized by sheaves on smooth manifolds, while
supermanifolds are constructed by gluing of sheaves on supervector
spaces. Treating odd variables on a smooth manifold $X$, we follow
the Serre--Swan theorem generalized to graded manifolds (Theorem
\ref{vv0}). It states that, if a graded commutative
$C^\infty(X)$-ring is generated by a projective
$C^\infty(X)$-module of finite rank, it is isomorphic to a ring of
graded functions on a graded manifold whose body is $X$. In
accordance with this theorem, we describe odd variables in terms
of graded manifolds \cite{lmp08,cmp04,book09,ijgmmp07}.

We consider a generic Lagrangian theory of even and odd variables
on an $n$-dimensional smooth real manifold $X$. It is phrased in
terms of the Grassmann-graded variational bicomplex (\ref{7})
\cite{barn,jmp05a,lmp08,cmp04,book09,ijgmmp07}. Graded Lagrangians
$L$ and Euler--Lagrange operators $\dl L$ are defined as elements
of terms $\cS^{0,n}_\infty[F;Y]$ and $\vr(\cS^{1,n}_\infty[F;Y])$
of this bicomplex, respectively. Cohomology of the
Grassmann-graded variational bicomplex (\ref{7}) (Theorems
\ref{v11} -- \ref{v11'}) defines a class of variationally trivial
graded Lagrangians (Theorem \ref{cmp26}) and results in the global
decomposition (\ref{g99}) of $dL$ (Theorem \ref{g103}), the first
variational formula (\ref{g107}) and the first Noether Theorem
\ref{j45}.

A problem is that any Euler--Lagrange operator satisfies Noether
identities, which therefore must be separated into the trivial and
non-trivial ones. These Noether identities obey first-stage
Noether identities, which in turn are subject to the second-stage
ones, and so on. Thus, there is a hierarchy of higher-stage
Noether identities. In accordance with general analysis of Noether
identities of differential operators \cite{oper}, if certain
conditions hold, one can associate to a graded Lagrangian system
the exact antifield Koszul--Tate complex (\ref{v94}) possessing
the boundary operator (\ref{v92}) whose nilpotentness is
equivalent to all non-trivial Noether and higher-stage Noether
identities \cite{jmp05a,lmp08,jmp09}.

It should be noted that the notion of higher-stage Noether
identities has come from that of reducible constraints. The
Koszul--Tate complex of Noether identities has been invented
similarly to that of constraints under the condition that Noether
identities are locally separated into independent and dependent
ones \cite{barn,fisch}. This condition is relevant for
constraints, defined by a finite set of functions which the
inverse mapping theorem is applied to. However, Noether identities
unlike constraints are differential equations. They are given by
an infinite set of functions on a Fr\'echet manifold of infinite
order jets where the inverse mapping theorem fails to be valid.
Therefore, the regularity condition for the Koszul--Tate complex
of constraints is replaced with homology regularity Condition
\ref{v155} in order to construct the Koszul--Tate complex
(\ref{v94}) of Noether identities.

The second Noether theorems (Theorems \ref{w35}, \ref{825} and
\ref{826}) is formulated in homology terms, and it associates to
this Koszul--Tate complex the cochain sequence of ghosts
(\ref{w108}) with the ascent operator (\ref{w108'}) whose
components are non-trivial gauge and higher-stage gauge symmetries
of Lagrangian theory.

\section{Variational bicomplex on fibre bundles}

Given a smooth fibre bundle $Y\to X$, the jet manifolds $J^rY$ of
its sections provide the conventional language of theory of
differential equations and differential operators on $Y\to X$
\cite{bry,kras}. Though we restrict our consideration to finite
order Lagrangian formalism, it is conveniently  formulated on an
infinite order jet manifold $J^\infty Y$ of $Y$ in terms of the
above mentioned variational bicomplex of differential forms on
$J^\infty Y$. However, different variants of a variational
sequence of finite jet order on jet manifolds $J^rY$ also are
considered \cite{and,kru90,vitolo07}.

\begin{rem} Smooth manifolds throughout
are assumed to be Hausdorff, second-countable and,
consequently, paracompact and locally compact, countable at
infinity. It is essential that a paracompact smooth manifold
admits the partition of unity by smooth functions. Given a
manifold $X$, its tangent and cotangent bundles $TX$ and $T^*X$
are endowed with bundle coordinates $(x^\la,\dot x^\la)$ and
$(x^\la,\dot x_\la)$ with respect to holonomic frames
$\{\dr_\la\}$ and $\{dx^\la\}$, respectively. By
$\La=(\la_1...\la_k)$, $|\La|=k$, $\la+\La=(\la\la_1...\la_k)$,
are denoted symmetric multi-indices. Summation over a multi-index
$\La$ means separate summation over each index $\la_i$.
\end{rem}

Let $Y\to X$ be a fibre bundle provided with bundle coordinates
$(x^\la,y^i)$. An $r$-order jet manifold $J^rY$ of its sections is
provided with the adapted coordinates $(x^\la, y^i,
y^i_\La)_{|\La|\leq r}$. These jet manifolds form an inverse
system
\mar{j1}\beq
Y\op\longleftarrow^\pi J^1Y \longleftarrow \cdots J^{r-1}Y
\op\longleftarrow^{\pi^r_{r-1}} J^rY\longleftarrow\cdots,
\label{j1}
\eeq
where $\pi^r_{r-1}$, $r>0$, are affine bundles. Its projective
limit $J^\infty Y$ is defined as a minimal set such that there
exist surjections
\mar{5.74}\beq
\pi^\infty: J^\infty Y\to X, \quad \pi^\infty_0: J^\infty Y\to Y,
\quad \quad \pi^\infty_k: J^\infty Y\to J^kY, \label{5.74}
\eeq
obeying the relations $\pi^\infty_r=\pi^k_r\circ\pi^\infty_k$ for
all admissible $k$ and $r<k$. One can think of elements of
$J^\infty Y$ as being infinite order jets of sections of $Y\to X$.

A set $J^\infty Y$ is provided with the coarsest topology such
that the surjections $\pi^\infty_r$ (\ref{5.74}) are continuous.
Its base consists of inverse images of open subsets of $J^rY$,
$r=0,\ldots$, under the maps $\pi^\infty_r$. With this topology,
$J^\infty Y$ is a paracompact Fr\'echet (complete metrizable)
manifold \cite{cmp04,book09,tak2}. It is called the infinite order
jet manifold. One can show that surjections $\pi^\infty_r$ are
open maps admitting local sections, i.e., $J^\infty Y\to J^rY$ are
continuous bundles. A bundle coordinate atlas
$\{U_Y,(x^\la,y^i)\}$ of $Y\to X$ provides $J^\infty Y$ with a
manifold coordinate atlas
\mar{j3}\beq
\{(\pi^\infty_0)^{-1}(U_Y), (x^\la, y^i_\La)\}_{0\leq|\La|},
\qquad {y'}^i_{\la+\La}=\frac{\dr x^\m}{\dr x'^\la}d_\m y'^i_\La,
\qquad d_\la = \dr_\la + \op\sum_{0\leq|\La|}
y^i_{\la+\La}\dr_i^\La. \label{j3}
\eeq

\begin{theo} \label{17t1} \mar{17t1}
A fibre bundle $Y$ is a strong deformation retract of an infinite
order jet manifold $J^\infty Y$ \cite{ander,jmp,book09}.
\end{theo}

\begin{cor} \label{17c1} \mar{17c1}
By virtue of the well-known Vietoris--Begle theorem \cite{bred},
there is an isomorphism
\mar{j19'}\beq
H^*(J^\infty Y;\Bbb R)=H^*(Y;\Bbb R) \label{j19'}
\eeq
between the cohomology of $J^\infty Y$ with coefficients in the
constant sheaf $\Bbb R$ and that of $Y$.
\end{cor}

The inverse sequence (\ref{j1}) of jet manifolds yields a direct
sequence
\mar{5.7}\beq
\cO^*(X)\op\longrightarrow^{\pi^*} \cO^*(Y)
\op\longrightarrow^{\pi^1_0{}^*} \cO_1^* \longrightarrow \cdots
\cO_{r-1}^*\op\longrightarrow^{\pi^r_{r-1}{}^*}
 \cO_r^* \longrightarrow\cdots \label{5.7}
\eeq
of differential graded algebras (henceforth DGAs) $\cO^*(X)$,
$\cO^*(Y)$, $\cO_r^*=\cO^*(J^rY)$ of exterior forms on $X$, $Y$
and jet manifolds $J^rY$, where $\pi^r_{r-1}{}^*$ are the
pull-back monomorphisms. Its direct limit $\cO^*_\infty$ consists
of all exterior forms on finite order jet manifolds modulo the
pull-back identification. It is a DGA which inherits operations of
an exterior differential $d$ and an exterior product $\w$ of DGAs
$\cO^*_r$.

\begin{theo} \label{j4} \mar{j4} The cohomology $H^*(\cO_\infty^*)$ of
the de Rham complex
\mar{5.13} \beq
0\longrightarrow \Bbb R\longrightarrow \cO^0_\infty
\op\longrightarrow^d\cO^1_\infty \op\longrightarrow^d \cdots
\label{5.13}
\eeq
of a DGA $\cO^*_\infty$ equals the de Rham cohomology $H^*_{\rm
DR}(Y)$ of a fibre bundle $Y$ \cite{and,bau,book09}.
\end{theo}

One can think of elements of $\cO_\infty^*$ as being differential
forms on an infinite order jet manifold $J^\infty Y$ as follows.
Let $\gO^*_r$ be a sheaf of germs of exterior forms on $J^rY$ and
$\ol\gO^*_r$ the canonical presheaf of local sections of
$\gO^*_r$. Since $\pi^r_{r-1}$ are open maps, there is a direct
sequence of presheaves
\be
\ol\gO^*_0 \op\longrightarrow^{\pi^1_0{}^*} \ol\gO_1^* \cdots
\op\longrightarrow^{\pi^r_{r-1}{}^*}
 \ol\gO_r^* \longrightarrow\cdots.
\ee
Its direct limit $\ol\gO^*_\infty$ is a presheaf of DGAs on
$J^\infty Y$. Let $\gQ^*_\infty$ be a sheaf of DGAs of germs of
$\ol\gO^*_\infty$ on $J^\infty Y$. The structure module
$\cQ^*_\infty=\G(\gQ^*_\infty)$ of global sections of
$\gQ^*_\infty$ is a DGA such that, given an element $\f\in
\cQ^*_\infty$ and a point $z\in J^\infty Y$, there exist an open
neighbourhood $U$ of $z$ and an exterior form $\f^{(k)}$ on some
finite order jet manifold $J^kY$ so that $\f|_U=
\pi^{\infty*}_k\f^{(k)}|_U$. Therefore, one can think of
$\cQ^*_\infty$ as being an algebra of locally exterior forms on
finite order jet manifolds. In particular, there is a monomorphism
$\cO^*_\infty \to\cQ^*_\infty$.

A DGA $\cO_\infty^*$ is split into a variational bicomplex
\cite{lmp00,jmp,cmp04,book09}. If $Y\to X$ is a contractible
bundle $\Bbb R^{n+p}\to \Bbb R^n$, a variational bicomplex is
exact \cite{olv,tul}. A problem is to determine cohomology of this
bicomplex in a general case. One also considers a variational
bicomplex of a DGA $\cQ^*_\infty$ \cite{ander,tak2}. It is
essential that a paracompact space $J^\infty Y$ admits a partition
of unity by elements of a ring $\cQ^0_\infty$ \cite{tak2}. This
fact enabled one to apply the abstract de Rham theorem (Theorem
\ref{+132}) in order to obtain cohomology of a variational
bicomplex $\cQ^*_\infty$ \cite{ander,tak2}. Then we have proved
that cohomology of a variational bicomplex $\cO_\infty^*$ equals
that of a variational bicomplex $\cQ^*_\infty$
\cite{lmp00,jmp,cmp04,book09,ijmms}.

\begin{rem} \label{polin} \mar{polin}
Let $Y\to X$ be a vector bundle. Its global section constitute a
projective $C^\infty(X)$-module of finite rank. The converse also
is true by virtue of the well-known Serre--Swan theorem, extended
to an arbitrary manifold $X$ \cite{book09,ren}. In this case, a
DGA $\cO^*_0$ of exterior forms on $Y$ is isomorphic to the
minimal Chevalley--Eilenberg differential calculus over  a real
commutative ring $C^\infty(Y)$ of smooth real functions on $Y$.
Jet bundles $J^rY\to X$ also are vector bundles. Then one can
consider a differential graded subalgebra $\cP^*_r\subset \cO^*_r$
of differential forms whose coefficients are polynomials in jet
coordinates $y^i_\La$, $0\leq |\La|\leq r$, on $J^r Y\to X$. In
particular, $\cP^0_r$ is a $C^\infty(X)$-ring of polynomials of
coordinates $y^i_\La$. One can associate to such a polynomial of
degree $m$ a section of a symmetric tensor product $\op\vee^m
(J^kY)^*$ of the dual of a jet bundle $J^kY\to X$, and {\it vice
versa}. A DGA $\cP^*_r$ is isomorphic to the minimal
Chevalley--Eilenberg differential calculus over a real ring
$\cP^0_r$. Accordingly, there exists a differential graded
subalgebra $\cP^*_\infty\subset \cO^*_\infty$ of differential
forms whose coefficients are polynomials in jet coordinates
$y^i_\La$, $0\leq |\La|$, of the continuous bundle $J^\infty Y\to
X$. This property is coordinate-independent due to the linear
transition functions (\ref{j3}). In particular, $\cP^0_\infty$ is
a ring of polynomials of coordinates $y^i_\La$, $0\leq |\La|$,
with coefficients in a ring $C^\infty(X)$. A DGA $\cP^*_\infty$ is
the direct system of the above mentioned DGAs $\cP^0_r$. It is
split into a variational bicomplex. Its cohomology can be obtained
\cite{lmp00,cmp04,book09,ijgmmp07}.
\end{rem}

We follow this example in order to construct a Grassmann-graded
variational bicomplex.

\section{Differential calculus over a graded commutative ring}

Let us start with the differential calculus over a graded
commutative ring (henceforth GCR) as a generalization of that over
a commutative ring.

By a Grassmann gradation (or, simply, a gradation if there is no
danger of confusion) throughout is meant a $\Bbb Z_2$-gradation.
Hereafter, the symbol $\nw .$ stands for a Grassmann parity.

An additive group $\cA$ is said to be graded if it is a product
$\cA=\cA_0\oplus \cA_1$ of two additive subgoups $\cA_0$ and
$\cA_1$ whose elements are called even and odd, respectively.

A an algebra $\cA$ is called graded if it is a graded additive
group so that
\be
[aa']=([a]+[a']){\rm mod}\,2, \qquad a\in \cA_{[a]}, \qquad a'\in
\cA_{[a']}.
\ee
Its even part $\cA_0$ is a subalgebra of $\cA$, while the odd one
$\cA_1$ is an $\cA$-module. If $\cA$ is a graded ring, then
$[1]=0$. A graded ring $\cA$ is called graded commutative if
$aa'=(-1)^{[a][a']}a'a$.

Given a graded algebra $\cA$, an $\cA$-module $Q$ is called graded
if it is a graded additive group such that
\be
[aq]=[qa]=([a]+[q]){\rm mod}\,2, \qquad a\in\cA,\quad q\in Q.
\ee
If $\cA$ is a GCR, a graded $\cA$-module $Q$ usually is assumed to
obey the condition $qa = (-1)^{[a][q]}aq$.

In particular, a graded $\Bbb R$-module $B=B_0\oplus B_1$ is
called the graded vector space. It is said to be
$(n,m)$-dimensional if $B_0=\Bbb R^n$, $B_1=\Bbb R^m$.

Let $\cK$ be a commutative ring. A graded algebra $\cA$ is said to
be a $\cK$-algebra if it is a $\cK$-module. For instance,  it is
called a real graded algebra if $\cK=\Bbb R$.

Let $\cA$ be a GCR. The following are standard constructions of
new graded $\cA$-modules from the old ones.

$\bullet$ A direct sum of graded modules is defined just as that
of modules over a commutative ring.

$\bullet$ A tensor product $P\ot Q$ of graded $\cA$-modules $P$
and $Q$ is an additive group generated by elements $p\ot q$, $p\in
P$, $q\in Q$, obeying relations
\be
&& (p+p')\ot q =p\ot q + p'\ot q, \qquad  p\ot(q+q')=p\ot q+p\ot q', \\
&&  ap\ot q=(-1)^{[p][a]}pa\ot q= (-1)^{[p][a]}p\ot aq,  \qquad a\in\cA.
\ee
In particular, a tensor algebra $\ot P$ of a graded $\cA$-module
$P$ is defined as that of a module over a commutative algebra. Its
quotient $\w P$ with respect to an ideal generated by elements
\be
p\ot p' + (-1)^{[p][p']}p'\ot p, \qquad p,p'\in P,
\ee
is a bigraded exterior algebra of a graded module $P$ with respect
to a graded exterior product
\be
p\w p' =- (-1)^{[p][p']}p'\w p.
\ee

$\bullet$ A morphism $\Phi:P\to Q$ of graded $\cA$-modules seen as
additive groups is said to be an even (resp. odd) morphism if
$\Phi$ preserves (resp. change) the Grassmann parity of all
graded-homogeneous elements of $P$ and if it obeys the relations
\be
\Phi(ap)=(-1)^{[\Phi][a]}a\Phi(p), \qquad p\in P, \qquad a\in\cA.
\ee
A morphism $\Phi:P\to Q$ of graded $\cA$-modules as additive
groups is called a graded $\cA$-module morphism if it is
represented by a sum of even and odd morphisms. A set
$\hm_\cA(P,Q)$ of graded morphisms of $P$ to $Q$ is naturally a
graded $\cA$-module. A graded $\cA$-module $P^*=\hm_\cA(P,\cA)$ is
called the dual of a graded $\cA$-module $P$.

A real graded algebra $\cG$ is called a Lie superalgebra if its
product $[.,.]$, called the Lie superbracket, obeys relations
\be
&& [\ve,\ve']=-(-1)^{[\ve][\ve']}[\ve',\ve],\\
&& (-1)^{[\ve][\ve'']}[\ve,[\ve',\ve'']]
+(-1)^{[\ve'][\ve]}[\ve',[\ve'',\ve]] +
(-1)^{[\ve''][\ve']}[\ve'',[\ve,\ve']] =0.
\ee
Obviously, an even part $\cG_0$ of a Lie superalgebra $\cG$ is a
Lie algebra. A graded vector space $P$ is called a $\cG$-module if
it is provided with an $\Bbb R$-bilinear map
\be
&& \cG\times P\ni (\ve,p)\to \ve p\in P, \qquad [\ve
p]=([\ve]+[p]){\rm mod}\,2,\\
&& [\ve,\ve']p=(\ve\circ\ve'-(-1)^{[\ve][\ve']}\ve'\circ\ve)p.
\ee

Let $\cA$ be a real GCR. Let $P$ and $Q$ be graded $\cA$-modules.
The real graded module $\hm_{\Bbb R}(P,Q)$ of $\Bbb R$-linear
graded homomorphisms $\Phi:P\to Q$ can be endowed with the two
graded $\cA$-module structures
\be
(a\Phi)(p)= a\Phi(p), \qquad (\Phi\bll a)(p) = \Phi (a p),\qquad
a\in \cA, \quad p\in P,
\ee
called $\cA$- and $\cA^\bll$-module structures, respectively. Let
us put
\be
\dl_a\Phi= a\Phi -(-1)^{[a][\Phi]}\Phi\bll a, \qquad a\in\cA.
\ee
An element $\Delta\in\hm_{\Bbb R}(P,Q)$ is said to be a $Q$-valued
graded differential operator of order $s$ on $P$ if
$\dl_{a_0}\circ\cdots\circ\dl_{a_s}\Delta=0$ for any tuple of
$s+1$ elements $a_0,\ldots,a_s$ of $\cA$.

In particular, zero order graded differential operators coincide
with graded $\cA$-module morphisms $P\to Q$. A first order graded
differential operator $\Delta$ satisfies the relation
\be
&& \dl_a\circ\dl_b\,\Delta(p)=
ab\Delta(p)- (-1)^{([b]+[\Delta])[a]}b\Delta(ap)-
(-1)^{[b][\Delta]}a\Delta(bp)+\\
&& \qquad (-1)^{[b][\Delta]+([\Delta]+[b])[a]}
=0, \qquad a,b\in\cA, \quad p\in P.
\ee

For instance, let $P=\cA$. A first order $Q$-valued graded
differential operator $\Delta$ on $\cA$ fulfils the condition
\be
\Delta(ab)= \Delta(a)b+ (-1)^{[a][\Delta]}a\Delta(b)
-(-1)^{([b]+[a])[\Delta]} ab \Delta(\bb), \qquad  a,b\in\cA.
\ee
It is called a $Q$-valued graded derivation of $\cA$ if
$\Delta(\bb)=0$, i.e., the graded Leibniz rule
\be
\Delta(ab) = \Delta(a)b + (-1)^{[a][\Delta]}a\Delta(b), \quad
a,b\in \cA,
\ee
holds. If $\dr$ is a graded derivation of $\cA$, then $a\dr$ is so
for any $a\in \cA$. Hence, graded derivations of $\cA$ constitute
a graded $\cA$-module $\gd(\cA,Q)$, called the graded derivation
module.

If $Q=\cA$, a graded derivation module $\gd\cA$ also is a real Lie
superalgebra with respect to a superbracket
\mar{ws14}\beq
[u,u']=u\circ u' - (-1)^{[u][u']}u'\circ u, \qquad u,u'\in \cA.
\label{ws14}
\eeq
Then one can consider the Chevalley--Eilenberg complex
$C^*[\gd\cA;\cA]$ where a real GCR $\cA$ is a regarded as an
$\gd\cA$-module \cite{fuks,book09}. It reads
\mar{ws85}\ben
&& 0\to \Bbb R\ar^{\rm in}\cA\ar^d C^1[\gd\cA;\cA]\ar^d \cdots
C^k[\gd\cA;\cA]\ar^d\cdots, \label{ws85}\\
&& C^k[\gd\cA;\cA]=\hm_{\Bbb R}(\op\w^k \gd\cA,\cA). \nonumber
\een
Let us bring homogeneous elements of $\op\w^k \gd\cA$ into the
form
\be
\ve_1\w\cdots\ve_r\w\e_{r+1}\w\cdots\w \e_k, \qquad
\ve_i\in(\gd\cA)_0, \quad \e_j\in(\gd\cA)_1.
\ee
Then an even coboundary operator $d$ of the complex (\ref{ws85})
is given by the expression
\mar{ws86}\ben
&& dc(\ve_1\w\cdots\w\ve_r\w\e_1\w\cdots\w\e_s)=
\label{ws86}\\
&&\op\sum_{i=1}^r (-1)^{i-1}\ve_i
c(\ve_1\w\cdots\wh\ve_i\cdots\w\ve_r\w\e_1\w\cdots\e_s)+
\nonumber \\
&& \op\sum_{j=1}^s (-1)^r\ve_i
c(\ve_1\w\cdots\w\ve_r\w\e_1\w\cdots\wh\e_j\cdots\w\e_s)
+\nonumber\\
&& \op\sum_{1\leq i<j\leq r} (-1)^{i+j}
c([\ve_i,\ve_j]\w\ve_1\w\cdots\wh\ve_i\cdots\wh\ve_j
\cdots\w\ve_r\w\e_1\w\cdots\w\e_s)+\nonumber\\
&&\op\sum_{1\leq i<j\leq s} c([\e_i,\e_j]\w\ve_1\w\cdots\w
\ve_r\w\e_1\w\cdots
\wh\e_i\cdots\wh\e_j\cdots\w\e_s)+\nonumber\\
&& \op\sum_{1\leq i<r,1\leq j\leq s} (-1)^{i+r+1}
c([\ve_i,\e_j]\w\ve_1\w\cdots\wh\ve_i\cdots\w\ve_r\w
\e_1\w\cdots\wh\e_j\cdots\w\e_s),\nonumber
\een
where the caret $\,\wh{}\,$ denotes omission.

It is easily justified that the complex (\ref{ws85}) contains a
subcomplex $\cO^*[\gd\cA]$ of $\cA$-linear graded morphisms. It is
provided with a structure of a bigraded $\cA$-algebra with respect
to a graded exterior product
\mar{ws103'}\beq
\f\w\f'(u_1,...,u_{r+s})= \op\sum_{i_1<\cdots<i_r;j_1<\cdots<j_s}
{\rm Sgn}^{i_1\cdots i_rj_1\cdots j_s}_{1\cdots r+s}
\f(u_{i_1},\ldots, u_{i_r}) \f'(u_{j_1},\ldots,u_{j_s}),
\label{ws103'}
\eeq
where $u_1,\ldots, u_{r+s}$ are graded-homogeneous elements of
$\gd\cA$ and
\be
u_1\w\cdots \w u_{r+s}= {\rm Sgn}^{i_1\cdots i_rj_1\cdots
j_s}_{1\cdots r+s} u_{i_1}\w\cdots\w u_{i_r}\w u_{j_1}\w\cdots\w
u_{j_s}.
\ee
The coboundary operator $d$ (\ref{ws86}) and the graded exterior
product $\w$ (\ref{ws103'}) bring $\cO^*[\gd\cA]$ into a
differential bigraded algebra (henceforth DBGA) whose elements
obey relations
\be
\f\w \f'=(-1)^{|\f||\f'|+[\f][\f']}\f'\w\f, \qquad d(\f\w\f')=
d\f\w\f' +(-1)^{|\f|}\f\w d\f'.
\ee
It is called the graded Chevalley--Eilenberg differential calculus
over a real GCR $\cA$.

In particular, we have
\mar{ws47}\beq
\cO^1[\gd\cA]=\hm_\cA(\gd\cA,\cA)=\gd\cA^*. \label{ws47}
\eeq
One can extend this duality relation to the graded interior
product of $u\in\gd\cA$ with any element $\f\in \cO^*[\gd\cA]$ by
the rules
\be
&& u\rfloor(bda) =(-1)^{[u][b]}bu(a),\qquad a,b \in\cA, \\
&& u\rfloor(\f\w\f')=
(u\rfloor\f)\w\f'+(-1)^{|\f|+[\f][u]}\f\w(u\rfloor\f').
\ee
As a consequence, any graded derivation $u\in\gd\cA$ of $\cA$
yields a derivation
\be
&& \bL_u\f= u\rfloor d\f + d(u\rfloor\f), \qquad \f\in\cO^*[\gd\cA], \qquad
u\in\gd\cA, \\
&& \bL_u(\f\w\f')=\bL_u(\f)\w\f' + (-1)^{[u][\f]}\f\w\bL_u(\f'),
\ee
called the graded Lie derivative of a DBGA $\cO^*[\gd\cA]$.

Note that, if $\cA$ is a commutative ring, the graded
Chevalley--Eilenberg differential calculus comes to the familiar
one.

The minimal graded Chevalley--Eilenberg differential calculus
$\cO^*\cA\subset \cO^*[\gd\cA]$  over a GCR $\cA$ consists of
monomials $a_0da_1\w\cdots\w da_k$, $a_i\in\cA$. The corresponding
complex
\mar{t100}\beq
0\to\Bbb R\ar \cA\ar^d\cO^1\cA\ar^d \cdots  \cO^k\cA\ar^d \cdots
\label{t100}
\eeq
is called the  bigraded de Rham complex of a real GCR $\cA$.

\section{Differential calculus on graded manifolds}

As was mentioned above, we follow Serre--Swan Theorem \ref{vv0}
below and consider a real GCR $\cA$ of graded functions on a
graded manifold. Then the minimal graded Chevalley--Eilenberg
differential calculus $\cO^*\cA$ over $\cA$ is a DBGA of graded
exterior forms on this graded manifold
\cite{cmp04,book09,ijgmmp07}.

A real GCR $\La$ is called the Grassmann algebra if it is a free
ring such that
\be
\La = \La_0 \oplus \La_1=(\Bbb R\oplus (\La_1)^2)\oplus \La_1,
\ee
i.e., a Grassmann algebra is generated by the unit element $\bb$
and its odd elements. Note that there is a different definition of
a Grassmann algebra \cite{jad}.

Hereafter, we restrict our consideration to Grassmann algebras
which are finite-dimensional vector spaces. In this case, there
exists a real vector space $V$ such that $\La=\w V$ is its
exterior algebra endowed with the Grassmann gradation
\mar{+66}\beq
\La_0=\Bbb R\op\bigoplus_{k=1} \op\w^{2k} V, \qquad
\La_1=\op\bigoplus_{k=1} \op\w^{2k-1} V. \label{+66}
\eeq
One calls $\di V$ the rank of a Grassmann algebra $\La$. Given a
basis $\{c^i\}$ for a vector space $V$, elements of the Grassmann
algebra $\La$ (\ref{+66}) take the form
\be
a=\op\sum_{k=0,1,\ldots} \op\sum_{(i_1\cdots i_k)}a_{i_1\cdots
i_k}c^{i_1}\cdots c^{i_k},
\ee
where the second sum runs through all the tuples $(i_1\cdots i_k)$
such that no two of them are permutations of each other.

A graded manifold of dimension $(n,m)$ is defined as a
local-ringed space $(Z,\gA)$ whose body $Z$ is an $n$-dimensional
smooth manifold and whose structure sheaf $\gA=\gA_0\oplus\gA_1$
is a sheaf of Grassmann algebras of rank $m$ such that
\cite{bart,book09}:

$\bullet$ there is an exact sequence of sheaves
\be
0\to \cR \to\gA \op\to^\si C^\infty_Z\to 0, \qquad
\cR=\gA_1+(\gA_1)^2,
\ee
where $C^\infty_Z$ is a sheaf of smooth real functions on $Z$;

$\bullet$ $\cR/\cR^2$ is a locally free sheaf of
$C^\infty_Z$-modules of finite rank (with respect to pointwise
operations), and a sheaf $\gA$ is locally isomorphic to an
exterior product $\w_{C^\infty_Z}(\cR/\cR^2)$.

Sections of a sheaf $\gA$ are called graded functions on a graded
manifold $(Z,\gA)$. They make up a real GCR $\gA(Z)$ which is a
$C^\infty(Z)$-ring, called the structure ring of $(Z,\gA)$. Let us
recall the well-known Batchelor theorem \cite{bart,book09}.

\begin{theo} \label{lmp1a} \mar{lmp1a}
Let $(Z,\gA)$ be a graded manifold. There exists a vector bundle
$E\to Z$ with an $m$-dimensional typical fibre $V$ such that the
structure sheaf $\gA$ of $(Z,\gA)$ is isomorphic to the structure
sheaf $\gA_E=S_{\w E^*}$ of germs of sections of an exterior
bundle $\w E^*$, whose typical fibre is a Grassmann algebra
$\La=\w V^*$.
\end{theo}

Though Batchelor's isomorphism in Theorem \ref{lmp1a} fails to be
canonical, we restrict our consideration to graded manifolds
$(Z,\gA_E)$, called simple graded manifolds modelled over a vector
bundle $E\to Z$. Accordingly, the structure ring $\cA_E$ of a
simple graded manifold $(Z,\gA_E)$ is a module $\cA_E=\w E^*(Z)$
of sections of an exterior bundle $\w E^*$.

The above-mentioned Serre--Swan theorem and Theorem \ref{lmp1a}
lead to the Serre--Swan theorem for graded manifolds
\cite{jmp05a,book09}.

\begin{theo} \label{vv0} \mar{vv0}
Let $Z$ be a smooth manifold. A $C^\infty(Z)$-GCR $\cA$ is
generated by some projective $C^\infty(Z)$-module of finite rank
iff it is isomorphic to the structure ring $\gA(Z)$ of some graded
manifold $(Z,\gA)$ with a body $Z$.
\end{theo}

Given a simple graded manifold $(Z,\gA_E)$, a trivialization chart
$(U; z^A,y^a)$ of a vector bundle $E\to Z$ yields its splitting
domain $(U; z^A,c^a)$. Graded functions on it are $\La$-valued
functions
\mar{z785}\beq
f=\op\sum_{k=0}^m \frac1{k!}f_{a_1\ldots a_k}(z)c^{a_1}\cdots
c^{a_k}, \label{z785}
\eeq
where $f_{a_1\cdots a_k}(z)$ are smooth functions on $U$ and
$\{c^a\}$ is a fibre basis for $E^*$. One calls $\{z^A,c^a\}$ the
local basis for a graded manifold $(Z,\gA_E)$ \cite{bart,book09}.
Transition functions $y'^a=\rho^a_b(z^A)y^b$ of bundle coordinates
on $E\to Z$ yield the corresponding transformation
$c'^a=\rho^a_b(z^A)c^b$ of the associated local basis for a graded
manifold $(Z,\gA_E)$ and the according coordinate transformation
law of graded functions (\ref{z785}).

Given a graded manifold $(Z,\gA)$, let $\gd\gA(Z)$ be a graded
derivation module of its real structure ring $\gA(Z)$. Its
elements are called graded vector fields on a graded manifold
$(Z,\gA)$. A key point is that graded vector fields $u\in\gd\cA_E$
on a simple graded manifold $(Z,\gA_E)$ can be represented by
sections of some vector bundle as follows \cite{cmp04,book09}. Due
to a canonical splitting $VE= E\times E$, the vertical tangent
bundle $VE$ of $E\to Z$ can be provided with fibre bases
$\{\dr_a\}$, which are the duals of bases $\{c^a\}$. Then graded
vector fields on a splitting domain $(U;z^A,c^a)$ of $(Z,\gA_E)$
read
\mar{hn14}\beq
u= u^A\dr_A + u^a\dr_a, \label{hn14}
\eeq
where $u^A, u^a$ are local $\La$-valued functions on $U$. In
particular,
\be
\dr_a\circ\dr_b =-\dr_b\circ\dr_a, \qquad
\dr_A\circ\dr_a=\dr_a\circ \dr_A.
\ee
The graded derivations (\ref{hn14}) act on graded functions
$f\in\gA_E(U)$ (\ref{z785}) by the rule
\mar{cmp50a}\beq
u(f_{a\ldots b}c^a\cdots c^b)=u^A\dr_A(f_{a\ldots b})c^a\cdots c^b
+u^k f_{a\ldots b}\dr_k\rfloor (c^a\cdots c^b). \label{cmp50a}
\eeq
This rule implies the corresponding coordinate transformation law
\be
u'^A =u^A, \qquad u'^a=\rho^a_ju^j +u^A\dr_A(\rho^a_j)c^j
\ee
of graded vector fields. It follows that graded vector fields
(\ref{hn14}) can be represented by sections of a vector bundle
$\cV_E$ which is locally isomorphic to a vector bundle $\w
E^*\op\ot_Z(E\op\oplus_Z TZ)$.

Given a real GCR $\cA_E$ of graded functions on a graded manifold
$(Z,\gA_E)$ and a real Lie superalgebra $\gd\cA_E$ of its graded
derivations, let us consider the graded Chevalley--Eilenberg
differential calculus
\mar{33f21}\beq
\cS^*[E;Z]=\cO^*[\gd\cA_E] \label{33f21}
\eeq
over $\cA_E$. Since a graded derivation module $\gd\cA_E$ is
isomorphic to a module of sections of a vector bundle $\cV_E\to
Z$, elements of $\cS^*[E;Z]$ are represented by sections of an
exterior bundle $\w\ol\cV_E$ of the $\w E^*$-dual $\ol\cV_E\to Z$
of $\cV_E$ which is locally isomorphic to a vector bundle $\w
E^*\op\ot_Z(E^*\op\oplus_Z T^*Z)$. With respect to the dual fibre
bases $\{dz^A\}$ for $T^*Z$ and $\{dc^b\}$ for $E^*$, sections of
$\ol\cV_E$ take the coordinate form
\be
\f=\f_A dz^A + \f_adc^a, \qquad \f'_a=\rho^{-1}{}_a^b\f_b, \qquad
\f'_A=\f_A +\rho^{-1}{}_a^b\dr_A(\rho^a_j)\f_bc^j,
\ee
$\f_A$, $f_a$ are local $\La$-valued functions on $U$. The duality
isomorphism $\cS^1[E;Z]=\gd\cA_E^*$ (\ref{ws47}) is given by a
graded interior product
\be
u\rfloor \f=u^A\f_A + (-1)^{\nw{\f_a}}u^a\f_a.
\ee
Elements of $\cS^*[E;Z]$ are called graded exterior forms on a
graded manifold $(Z,\gA_E)$.

Seen as an $\cA_E$-algebra, the DBGA $\cS^*[E;Z]$ (\ref{33f21}) on
a splitting domain $(U;z^A,c^a)$ is locally generated by graded
one-forms $dz^A$, $dc^i$ such that
\be
dz^A\w dc^i=-dc^i\w dz^A, \qquad dc^i\w dc^j= dc^j\w dc^i.
\ee
Accordingly, the coboundary operator $d$ (\ref{ws86}), called the
graded exterior differential, reads
\be
d\f= dz^A \w \dr_A\f +dc^a\w \dr_a\f,
\ee
where derivatives $\dr_A$, $\dr_a$ act on coefficients of graded
exterior forms by the formula (\ref{cmp50a}), and they are graded
commutative with graded forms $dz^A$, $dc^a$.

\begin{lem} \label{v62} \mar{v62}
The DBGA $\cS^*[E;Z]$ (\ref{33f21}) is a minimal differential
calculus over $\cA_E$ \cite{book09}.
\end{lem}

The bigraded de Rham complex (\ref{t100}) of the minimal graded
Chevalley--Eilenberg differential calculus $\cS^*[E;Z]$ reads
\mar{+137}\beq
0\to \Bbb R\to \cA_E \ar^d \cS^1[E;Z]\ar^d\cdots
\cS^k[E;Z]\ar^d\cdots. \label{+137}
\eeq
Its cohomology $H^*(\cA_E)$  is called the de Rham cohomology of a
graded manifold $(Z,\gA_E)$. In particular, given a DGA $\cO^*(Z)$
of exterior forms on $Z$, there exists a canonical monomorphism
\mar{uut}\beq
\cO^*(Z)\to \cS^*[E;Z] \label{uut}
\eeq
and a body epimorphism $\cS^*[E;Z]\to \cO^*(Z)$ which are cochain
morphisms of the de Rham complex (\ref{+137}) and the de Rham
complex of $\cO^*(Z)$. Then one can show the following
\cite{book09,ijgmmp07}.

\begin{theo} \label{33t3} \mar{33t3}
The de Rham cohomology of a graded manifold $(Z,\gA_E)$ equals the
de Rham cohomology of its body $Z$.
\end{theo}

\begin{cor}
Any closed graded exterior form is decomposed into a sum $\f=\si
+d\xi$ where $\si$ is a closed exterior form on $Z$.
\end{cor}

\section{Grassmann-graded variational bicomplex}

Let $X$ be an $n$-dimensional smooth manifold and $Y\to X$ a
vector bundle over $X$. In Remark \ref{polin}, we mention a
polynomial variational bicomplex of a DGA $\cP^*_\infty$. The
latter is the direct limit of DGAs $\cP^*_r$ where $\cP^*_r$ is
the minimal Chevalley--Eilenberg differential calculus over a ring
$\cP^0_r$ of sections of symmetric tensor products of a vector jet
bundle $J^rY\to X$.

Let $(X,\gA_E)$ be a simple graded manifold modelled over a vector
bundle $E\to X$. Any jet bundle $J^rE\to X$ also is a vector
bundle. Then let $(X,\cA_{J^rE})$ denote a simple graded manifold
modelled over a vector bundle $J^rE\to X$. Its structure module
$\cA_{J^rE}$ is a real GCR of sections of an exterior bundle $\w
(J^rE)^*$ where $(J^rE)^*$ denotes the dual of $J^rE\to X$. Let
$S^*[J^rE,X]$ be the minimal Chevalley--Eilenberg differential
calculus over a real GCR $\cA_{J^rE}$. It is a BGDA of graded
exterior forms on a simple graded manifold $(X,\cA_{J^rE})$. There
is a direct system
\be
\cS^*[E;X]\ar \cS^*[J^1E;X]\ar\cdots \cS^*[J^rE;X]\ar\cdots
\ee
of BGDAs $S^*[J^rE,X]$. Its direct limit $\cS^*_\infty[E;X]$  is
the Grassmann-graded counterpart of the above mentioned DGA
$\cP^*_\infty$. A BDGA $\cS^*_\infty[E;X]$ is split into a
Grassmann-graded variational bicomplex which leads to graded
Lagrangian formalism of odd variables represented by generating
elements of the structure ring $\cA_E$ of a graded manifold
$(X,\gA_E)$ \cite{cmp04,book09,ijgmmp07}.

Note that the definition of jets of these odd variables as
elements of structure rings of graded manifolds $\cA_{J^rE}$
differs from that of jets of fibred graded manifolds
\cite{hern,mont06}, but it reproduces the heuristic notion of jets
of odd variables in Lagrangian field theory \cite{barn,bran01}.

In order to formulate graded Lagrangian theory both of even and
odd variables, let us consider a composite bundle $F\to Y\to X$
where $F\to Y$ is a vector bundle provided with bundle coordinates
$(x^\la, y^i, q^a)$. Jet manifolds $J^rF$ of $F\to X$ also are
vector bundles $J^rF\to J^rY$ coordinated by $(x^\la, y^i_\La,
q^a_\La)$, $0\leq |\La|\leq r$. Let $(J^rY,\gA_r)$ (where
$J^0Y=Y$, $\gA_0=\gA_F$) be a simple graded manifold modelled over
such a vector bundle. Its local basis is $(x^\la, y^i_\La,
c^a_\La)$, $0\leq|\La|\leq r$. Let
$\cS^*_r[F;Y]=\cS^*_r[J^rF;J^rY]$ denote a DBGA of graded exterior
forms on a graded manifold $(J^rY,\gA_r)$. In particular, there is
the cochain monomorphism (\ref{uut}):
\mar{34f3}\beq
\cO^*_r=\cO^*(J^rY)\to \cS^*_r[F;Y]. \label{34f3}
\eeq

A surjection $\pi^{r+1}_r:J^{r+1}Y\to J^rY$ yields an epimorphism
of graded manifolds
\be
(\pi^{r+1}_r,\wh \pi^{r+1}_r):(J^{r+1}Y,\gA_{r+1}) \to
(J^rY,\gA_r),
\ee
including a sheaf monomorphism $\wh
\pi^{r+1}_r:\pi_r^{r+1*}\gA_r\to \gA_{r+1}$, where
$\pi_r^{r+1*}\gA_r$ is the pull-back onto $J^{r+1}Y$ of a
continuous fibre bundle $\gA_r\to J^rY$. This sheaf monomorphism
induces a monomorphism of canonical presheaves $\ol \gA_r\to \ol
\gA_{r+1}$, which associates to each open subset $U\subset
J^{r+1}Y$ a ring of sections of $\gA_r$ over $\pi^{r+1}_r(U)$.
Accordingly, there is a monomorphism
\mar{34f1}\beq
\pi_r^{r+1*}:\cS^0_r[F;Y]\to \cS^0_{r+1}[F;Y] \label{34f1}
\eeq
of structure rings of graded functions on graded manifolds
$(J^rY,\gA_r)$ and $(J^{r+1}Y,\gA_{r+1})$. By virtue of Lemma
\ref{v62}, the differential calculus $\cS^*_r[F;Y]$ and
$\cS^*_{r+1}[F;Y]$ are minimal. Therefore, the monomorphism
(\ref{34f1}) yields a monomorphism of DBGAs
\mar{v4'}\beq
\pi_r^{r+1*}:\cS^*_r[F;Y]\to \cS^*_{r+1}[F;Y]. \label{v4'}
\eeq
As a consequence, we have a direct system of DBGAs
\mar{j2}\beq
\cS^*[F;Y]\ar^{\pi^*} \cS^*_1[F;Y]\ar\cdots \cS^*_{r-1}[F;Y]
\op\ar^{\pi^{r*}_{r-1}}\cS^*_r[F;Y]\ar\cdots. \label{j2}
\eeq
Its direct limit $\cS^*_\infty [F;Y]$ consists of all graded
exterior forms $\f\in \cS^*[F_r;J^rY]$ on graded manifolds
$(J^rY,\gA_r)$ modulo the monomorphisms (\ref{v4'}).

The cochain monomorphisms $\cO^*_r\to \cS^*_r[F;Y]$ (\ref{34f3})
provide a monomorphism of the direct system (\ref{5.7}) to the
direct system (\ref{j2}) and, consequently, a monomorphism
\mar{v7}\beq
\cO^*_\infty\to \cS^*_\infty[F;Y] \label{v7}
\eeq
of their direct limits. In particular, $\cS^*_\infty[F;Y]$ is an
$\cO^0_\infty$-algebra. Accordingly, the body epimorphisms
$\cS^*_r[F;Y]\to \cO^*_r$ yield an epimorphism of
$\cO^0_\infty$-algebras
\mar{v7'}\beq
\cS^*_\infty[F;Y]\to \cO^*_\infty.  \label{v7'}
\eeq
It is readily observed that the morphisms (\ref{v7}) and
(\ref{v7'}) are cochain morphisms between the de Rham complex
(\ref{5.13}) of a DGA $\cO^*_\infty$ and the de Rham complex
\mar{g110}\beq
0\to\Bbb R\ar \cS^0_\infty[F;Y]\ar^d \cS^1_\infty[F;Y]\cdots
\ar^d\cS^k_\infty[F;Y] \ar\cdots \label{g110}
\eeq
of a DBGA $\cS^*_\infty[F;Y]$. Moreover, the corresponding
homomorphisms of cohomology groups of these complexes are
isomorphisms as follows.

\begin{theo} \label{v9} \mar{v9} There is an isomorphism
\mar{v10'}\beq
H^*(\cS^*_\infty[F;Y])= H^*_{DR}(Y) \label{v10'}
\eeq
of the cohomology of the de Rham complex (\ref{g110}) to the de
Rham cohomology of $Y$.
\end{theo}

\begin{proof}
The complex (\ref{g110}) is the direct limit of the de Rham
complexes of DBGAs $\cS^*_r[F;Y]$. In accordance with the
well-known theorem \cite{book09,massey}, the direct limit of
cohomology groups of these complexes is the cohomology of the de
Rham complex (\ref{g110}). By virtue of Theorem \ref{33t3},
cohomology of the de Rham complex of $\cS^*_r[F;Y]$ equals the de
Rham cohomology of $J^rY$ and, consequently, that of $Y$, which is
the strong deformation retract of any jet manifold $J^rY$ because
$J^kY\to J^{k-1}Y$ are affine bundles. Hence, the isomorphism
(\ref{v10'}) holds.
\end{proof}

\begin{cor} \mar{34c1} \label{34c1}
Any closed graded form $\f\in \cS^*_\infty[F;Y]$ is decomposed
into the sum $\f=\si +d\xi$ where $\si$ is a closed exterior form
on $Y$.
\end{cor}

One can think of  elements of $\cS^*_\infty[F;Y]$ as being graded
differential forms on an infinite order jet manifold $J^\infty Y$.
Indeed, let $\gS^*_r[F;Y]$ be a sheaf of DBGAs on $J^rY$ and
$\ol\gS^*_r[F;Y]$ its canonical presheaf. Then the above mentioned
presheaf monomorphisms $\ol \gA_r\to \ol \gA_{r+1}$ yield a direct
system of presheaves
\be
\ol\gS^*[F;Y]\ar \ol\gS^*_1[F;Y] \ar\cdots \ol\gS^*_r[F;Y]
\ar\cdots,
\ee
whose direct limit $\ol\gS_\infty^*[F;Y]$ is a presheaf of DBGAs
on an infinite order jet manifold $J^\infty Y$. Let
$\gQ^*_\infty[F;Y]$ be a sheaf of DBGAs of germs of a presheaf
$\ol\gS_\infty^*[F;Y]$. One can think of a pair $(J^\infty Y,
\gQ^0_\infty[F;Y])$ as being a graded Fr\'echet manifold, whose
body is an infinite order jet manifold $J^\infty Y$ and a
structure sheaf $\gQ^0_\infty[F;Y]$ is a sheaf of germs of graded
functions on graded manifolds $(J^rY,\gA_r)$. The structure module
$\cQ^*_\infty[F;Y]=\G(\gQ^*_\infty[F;Y])$ of sections of
$\gQ^*_\infty[F;Y]$ is a DBGA such that, given an element $\f\in
\cQ^*_\infty[F;Y]$ and a point $z\in J^\infty Y$, there exist an
open neighbourhood $U$ of $z$ and a graded exterior form
$\f^{(k)}$ on some finite order jet manifold $J^kY$ so that
$\f|_U= \pi^{\infty*}_k\f^{(k)}|_U$.

In particular, there is a monomorphism $\cS^*_\infty[F;Y]
\to\cQ^*_\infty[F;Y]$. Due to this monomorphism, one can restrict
$\cS^*_\infty[F;Y]$ to the coordinate chart (\ref{j3}) of
$J^\infty Y$ and can say that $\cS^*_\infty[F;Y]$ as an
$\cO^0_\infty$-algebra is locally generated by elements
\be
(c^a_\La,
dx^\la,\theta^a_\La=dc^a_\La-c^a_{\la+\La}dx^\la,\theta^i_\La=
dy^i_\La-y^i_{\la+\La}dx^\la), \qquad 0\leq |\La|,
\ee
where $c^a_\La$, $\theta^a_\La$ are odd and $dx^\la$,
$\theta^i_\La$ are even. We agree to call $(y^i,c^a)$ the local
generating basis for $\cS^*_\infty[F;Y]$. Let the collective
symbol $s^A$ stand for its elements. Accordingly, the notations
$s^A_\La$ of their jets and $\theta^A_\La=ds^A_\La-
s^A_{\la+\La}dx^\la$ of contact forms are introduced. For the sake
of simplicity, we further denote $[A]=[s^A]$.

A DBGA $\cS^*_\infty[F;Y]$ is decomposed into
$\cS^0_\infty[F;Y]$-modules $\cS^{k,r}_\infty[F;Y]$ of $k$-contact
and $r$-horizontal graded forms together with the corresponding
projections
\be
h_k:\cS^*_\infty[F;Y]\to \cS^{k,*}_\infty[F;Y], \qquad
h^m:\cS^*_\infty[F;Y]\to \cS^{*,m}_\infty[F;Y].
\ee
Accordingly, a graded exterior differential $d$ on
$\cS^*_\infty[F;Y]$ falls into the sum $d=d_V+d_H$ of a vertical
graded differential
\be
d_V \circ h^m=h^m\circ d\circ h^m, \qquad d_V(\f)=\theta^A_\La \w
\dr^\La_A\f, \qquad \f\in\cS^*_\infty[F;Y],
\ee
and a total graded differential
\be
&& d_H\circ h_k=h_k\circ d\circ h_k, \qquad d_H\circ h_0=h_0\circ d,
\qquad d_H(\f)=dx^\la\w d_\la(\f),\\
&& d_\la = \dr_\la + \op\sum_{0\leq|\La|} s^A_{\la+\La}\dr_A^\La.
\ee
These differentials obey the nilpotent relations
\be
d_H^2=0, \qquad d_V^2=0, \qquad d_Hd_V + d_Vd_H=0.
\ee

A DBGA $\cS^*_\infty[F;Y]$ also is provided with a graded
projection endomorphism
\be
&& \vr=\op\sum_{k>0} \frac1k\ol\vr\circ h_k\circ h^n:
\cS^{*>0,n}_\infty[F;Y]\to \cS^{*>0,n}_\infty[F;Y], \\
&& \ol\vr(\f)= \op\sum_{0\leq|\La|} (-1)^{\nm\La}\theta^A\w
[d_\La(\dr^\La_A\rfloor\f)], \qquad \f\in \cS^{>0,n}_\infty[F;Y],
\ee
such that $\vr\circ d_H=0$, and with a nilpotent graded
variational operator
\be
\dl=\vr\circ d: \cS^{*,n}_\infty[F;Y]\to \cS^{*+1,n}_\infty[F;Y].
\ee
These operators split a DBGA $\cS^{*,}_\infty[F;Y]$ into a
Grassmann-graded variational bicomplex
\mar{7}\beq
\begin{array}{ccccrlcrlcccrlcrl}
 & &  &  & & \vdots & & & \vdots  & & & & &
\vdots & &   & \vdots \\
& & & & _{d_V} & \put(0,-7){\vector(0,1){14}} & & _{d_V} &
\put(0,-7){\vector(0,1){14}} & &  & & _{d_V} &
\put(0,-7){\vector(0,1){14}}& & _{-\dl} & \put(0,-7){\vector(0,1){14}} \\
 &  & 0 & \to & &\cS^{1,0}_\infty[F;Y] &\op\to^{d_H} & &
\cS^{1,1}_\infty[F;Y] & \op\to^{d_H} & \cdots & & &
\cS^{1,n}_\infty[F;Y] &\op\to^\vr &  & \vr(\cS^{1,n}_\infty[F;Y])\to  0\\
& & & & _{d_V} &\put(0,-7){\vector(0,1){14}} & & _{d_V} &
\put(0,-7){\vector(0,1){14}} & & &  & _{d_V} &
\put(0,-7){\vector(0,1){14}}
 & & _{-\dl} & \put(0,-7){\vector(0,1){14}} \\
0 & \to & \Bbb R & \to & & \cS^0_\infty[F;Y] &\op\to^{d_H} & &
\cS^{0,1}_\infty[F;Y] & \op\to^{d_H} & \cdots & & &
\cS^{0,n}_\infty[F;Y] & \equiv &  & \cS^{0,n}_\infty[F;Y] \\
& & & & & \put(0,-7){\vector(0,1){14}} & &  &
\put(0,-7){\vector(0,1){14}} & & & &  &
\put(0,-7){\vector(0,1){14}} & &  & \\
0 & \to & \Bbb R & \to & & \cO^0(X) &\op\to^d & & \cO^1(X) &
\op\to^d & \cdots & & &
\cO^n(X) & \op\to^d & 0 &  \\
& & & & &\put(0,-5){\vector(0,1){10}} & & &
\put(0,-5){\vector(0,1){10}} & &  &  & &
\put(0,-5){\vector(0,1){10}} & &  & \\
& & & & &0 & &  & 0 & & & & &  0 & &  &
\end{array}
\label{7}
\eeq

We restrict our consideration to its short variational subcomplex
\mar{g111}\beq
 0\to \Bbb R\to \cS^0_\infty[F;Y]\ar^{d_H}\cS^{0,1}_\infty[F;Y]
\cdots \ar^{d_H} \cS^{0,n}_\infty[F;Y]\ar^\dl
\vr(\cS^{1,n}_\infty[F;Y]) \label{g111}
\eeq
and a subcomplex of one-contact graded forms
\mar{g112}\beq
 0\to \cS^{1,0}_\infty[F;Y]\ar^{d_H} \cS^{1,1}_\infty[F;Y]\cdots
\ar^{d_H}\cS^{1,n}_\infty[F;Y]\ar^\vr \vr(\cS^{1,n}_\infty[F;Y])
\to 0. \label{g112}
\eeq

\begin{theo} \label{v11} \mar{v11}
Cohomology of the complex (\ref{g111}) equals the de Rham
cohomology of $Y$.
\end{theo}

\begin{theo} \label{v11'} \mar{v11'}
The complex (\ref{g112}) is exact.
\end{theo}

These theorems are proved in Appendix B.

\section{Graded Lagrangian formalism}

Decomposed into the variational bicomplex, a DBGA
$\cS^*_\infty[F;Y]$ describes graded Lagrangian theory on a graded
manifold $(Y,\gA_F)$. Its graded Lagrangian is defined as an
element
\be
L=\cL\om\in \cS^{0,n}_\infty[F;Y], \qquad \om=dx^1\w\cdots\w dx^n,
\ee
of the graded variational complex (\ref{g111}). Accordingly, a
graded exterior form
\mar{0709'}\beq
\dl L= \theta^A\w \cE_A\om=\op\sum_{0\leq|\La|}
 (-1)^{|\La|}\theta^A\w d_\La (\dr^\La_A \cL)\om\in \vr(\cS^{1,n}_\infty[F;Y]) \label{0709'}
\eeq
is said to be its graded Euler--Lagrange operator. We agree to
call a pair $(\cS^{0,n}_\infty[F;Y],L)$ the graded Lagrangian
system.

The following is a corollary of Theorems \ref{v11} and
\ref{cmp26'} \cite{cmp04,book09}.

\begin{theo} \label{cmp26} \mar{cmp26}
Every $d_H$-closed graded form $\f\in\cS^{0,m<n}_\infty[F;Y]$
falls into the sum
\be
\f=h_0\si + d_H\xi, \qquad \xi\in \cS^{0,m-1}_\infty[F;Y],
\ee
where $\si$ is a closed $m$-form on $Y$. Any $\dl$-closed (i.e.,
variationally trivial) graded Lagrangian $L\in
\cS^{0,n}_\infty[F;Y]$ is the sum
\be
L=h_0\si + d_H\xi, \qquad \xi\in \cS^{0,n-1}_\infty[F;Y],
\ee
where $\si$ is a closed $n$-form on $Y$.
\end{theo}

The exactness of the complex (\ref{g112}) at a term
$\cS^{1,n}_\infty[F;Y]$ results in the following
\cite{cmp04,book09}.

\begin{theo} \label{g103} \mar{g103}
Given a graded Lagrangian $L$, there is the decomposition
\mar{g99,'}\ben
&& dL=\dl L - d_H\Xi_L,
\qquad \Xi\in \cS^{n-1}_\infty[F;Y], \label{g99}\\
&& \Xi_L=L+\op\sum_{s=0} \theta^A_{\nu_s\ldots\nu_1}\w
F^{\la\nu_s\ldots\nu_1}_A\om_\la, \label{g99'} \qquad
\om_\la=\dr_\la\rfloor\om,\\
&& F_A^{\nu_k\ldots\nu_1}= \dr_A^{\nu_k\ldots\nu_1}\cL-d_\la
F_A^{\la\nu_k\ldots\nu_1} +\si_A^{\nu_k\ldots\nu_1},\qquad
k=1,2,\ldots,\nonumber
\een
where local graded functions $\si$ obey the relations
$\si^\nu_A=0$, $\si_A^{(\nu_k\nu_{k-1})\ldots\nu_1}=0$.
\end{theo}

\begin{proof} The decomposition (\ref{g99}) is a
straightforward consequence of the exactness of the complex
(\ref{g112}) at a term $\cS^{1,n}_\infty[F,Y]$ and the fact that
$\vr$ is a projector. The coordinate expression (\ref{g99'})
results from a direct computation
\be
&& -d_H\Xi= -d_H[\th^A F^\la_A+\th^A_\nu F^{\la\nu}_A +\cdots
+\th^A_{\nu_s\ldots\nu_1} F^{\la\nu_s\ldots\nu_1}_A \\
&& \qquad +\th^A_{\nu_{s+1}\nu_s\ldots\nu_1}\w
F^{\la\nu_{s+1}\nu_s\ldots\nu_1}_A+ \cdots]\w\om_\la= [\th^Ad_\la
F^\la_A +\th^A_\nu (F^\nu_A +d_\la F^{\la\nu}_A)+\cdots \\
&& \qquad +\th^A_{\nu_{s+1}\nu_s\ldots\nu_1}
(F^{\nu_{s+1}\nu_s\ldots\nu_1}_A +
d_\la F^{\la\nu_{s+1}\nu_s\ldots\nu_1}_A) +\cdots]\w\om=\\
&&  \qquad [\th^Ad_\la F^\la_A +\th^A_\nu (\dr^\nu_A\cL)+\cdots
+ \th^A_{\nu_{s+1}\nu_s\ldots\nu_1}
(\dr^{\nu_{s+1}\nu_s\ldots\nu_1}_A\cL)+\cdots]\w\om=\\
&& \qquad  \th^A
(d_\la F^\la_A-\dr_A\cL) \w\om + dL= -\dl L+dL.
\ee
\end{proof}

The form $\Xi_L$ (\ref{g99'}) provides a global Lepage equivalent
of a graded Lagrangian $L$.

Given a graded Lagrangian system $(\cS^*_\infty[F;Y], L)$, by its
infinitesimal transformations are meant contact graded derivations
of a real GCR $\cS^0_\infty[F;Y]$. They constitute a
$\cS^0_\infty[F;Y]$-module $\gd \cS^0_\infty[F;Y]$ which is a real
Lie superalgebra with respect to the Lie superbracket
(\ref{ws14}). The following holds \cite{cmp04,book09}.

\begin{theo} \label{35t1} \mar{35t1}
The derivation module $\gd\cS^0_\infty[F;Y]$ is isomorphic to the
$\cS^0_\infty[F;Y]$-dual $(\cS^1_\infty[F;Y])^*$ of a module of
graded one-forms $\cS^1_\infty[F;Y]$. It follows that a DBGA
$\cS^*_\infty[F;Y]$ is the minimal Chevalley--Eilenberg
differential calculus over a real GCR $\cS^0_\infty[F;Y]$.
\end{theo}

Let $\vt\rfloor\f$, $\vt\in \gd\cS^0_\infty[F;Y]$, $\f\in
\cS^1_\infty[F;Y]$, denote the corresponding interior product.
Extended to a DBGA $\cS^*_\infty[F;Y]$, it obeys the rule
\be
\vt\rfloor(\f\w\si)=(\vt\rfloor \f)\w\si
+(-1)^{|\f|+[\f][\vt]}\f\w(\vt\rfloor\si), \qquad \f,\si\in
\cS^*_\infty[F;Y].
\ee

Restricted to the coordinate chart (\ref{j3}) of $J^\infty Y$, the
algebra $\cS^*_\infty[F;Y]$ is a free $\cS^0_\infty[F;Y]$-module
generated by one-forms $dx^\la$, $\theta^A_\La$. Due to the
isomorphism stated in Theorem \ref{35t1}, any graded derivation
$\vt\in\gd\cS^0_\infty[F;Y]$ takes the local form
\mar{gg3}\beq
\vt=\vt^\la \dr_\la + \vt^A\dr_A + \op\sum_{0<|\La|}\vt^A_\La
\dr^\La_A, \label{gg3}
\eeq
where $\dr^\La_A\rfloor dy_\Si^B=\dl_A^B\dl^\La_\Si$ up to
permutations of multi-indices $\La$ and $\Si$. Every graded
derivation $\vt$ (\ref{gg3}) yields a graded Lie derivative
\be
\bL_\vt\f=\vt\rfloor d\f+ d(\vt\rfloor\f), \qquad
\bL_\vt(\f\w\si)=\bL_\vt(\f)\w\si
+(-1)^{[\vt][\f]}\f\w\bL_\vt(\si),
\ee
of a DBGA $\cS^*_\infty[F;Y]$.

A graded derivation $\vt$ (\ref{gg3}) is called contact if a Lie
derivative $\bL_\vt$ preserves an ideal of contact graded forms of
a DBGA $\cS^*_\infty[F;Y]$. It takes the form
\mar{g105}\beq
\vt=\up_H+\up_V=\up^\la d_\la + [\up^A\dr_A +\op\sum_{|\La|>0}
d_\La(\up^A-s^A_\m\up^\m)\dr_A^\La], \label{g105}
\eeq
where $\up_H$ and $\up_V$ denotes the horizontal and vertical
parts of $\vt$ \cite{cmp04,book09}. A glance at the expression
(\ref{g105}) shows that a contact graded derivation $\vt$ as an
infinite order jet prolongation of its restriction
\mar{jj15}\beq
\up=\up^\la\dr_\la +\up^A\dr_A \label{jj15}
\eeq
to a GCR $S^0[F;Y]$. Since coefficients $\vt^\la$ and $\vt^i$
depend on jet coordinates $y^i_\La$, $0<|\La|$, in general, one
calls $\up$ (\ref{jj15}) a generalized vector field.

\begin{theo} \label{j44} \mar{j44}
A corollary of the decomposition (\ref{g99}) is that the Lie
derivative of a graded Lagrangian along any contact graded
derivation (\ref{g105}) obeys the first variational formula
\mar{g107}\beq
\bL_\vt L= \up_V\rfloor\dl L +d_H(h_0(\vt\rfloor \Xi_L)) + d_V
(\up_H\rfloor\om)\cL, \label{g107}
\eeq
where $\Xi_L$ is the Lepage equivalent (\ref{g99'}) of $L$
\cite{jmp05,cmp04}.
\end{theo}

A contact graded derivation $\vt$ (\ref{g105}) is called a
variational symmetry of a graded Lagrangian $L$ if the Lie
derivative $\bL_\vt L$ is $d_H$-exact, i.e., $\bL_\vt L=d_H\si$.

\begin{lem} \mar{35l10} \label{35l10}
A glance at the expression (\ref{g107}) shows the following. (i) A
contact graded derivation $\vt$ is a variational symmetry only if
it is projected onto $X$. (ii) Any projectable contact graded
derivation is a variational symmetry of a variationally trivial
graded Lagrangian. (iii) A contact graded derivations $\vt$ is a
variational symmetry iff its vertical part $\up_V$ (\ref{g105}) is
well. (iv) It is a variational symmetry iff a graded density
$\up_V\rfloor \dl L$ is $d_H$-exact.
\end{lem}

Note that generalized symmetries of differential equations and
Lagrangians of even variables has been intensively studied
\cite{and93,kras,olv}.

\begin{theo} \label{j45} \mar{j45} If a contact graded derivation $\vt$
(\ref{g105}) is a variational symmetry of a graded Lagrangian $L$,
the first variational formula (\ref{g107}) restricted to Ker$\,\dl
L$ leads to the weak conservation law
\be
 0\ap
d_H(h_0(\vt\rfloor\Xi_L)-\si).
\ee
\end{theo}

For the sake of brevity, the common symbol $\up$ further stands
for the generalized graded vector field $\up$ (\ref{jj15}), a
contact graded derivation $\vt$ determined by $\up$, and a Lie
derivative $\bL_\vt$.

A vertical contact graded derivation $\up= \up^A\dr_A$ is said to
be nilpotent if $\up(\up\f)=0$ for any horizontal graded form
$\f\in S^{0,*}_\infty[F,Y]$. It is nilpotent only if it is odd and
iff the equality $\up(\up^A)=0$ holds for all $\up^A$
\cite{cmp04}.

\begin{rem} \label{rr35} \mar{rr5}
For the sake of convenience, right derivations $\op\up^\lto
=\rdr_A\up^A$ also are  considered. They act on graded functions
and differential forms $\f$ on the right by the rules
\be
&& \op\up^\lto(\f)=d\f\lfloor\op\up^\lto +d(\f\lfloor\op\up^\lto),
\qquad \op\up^\lto(\f\w\f')=(-1)^{[\f']}\op\up^\lto(\f)\w\f'+
\f\w\op\up^\lto(\f'),\\
&& \theta_{\La A}\lfloor\rdr^{\Si B}=\dl^A_B\dl^\Si_\La.
\ee
\end{rem}

\section{Noether identities}

Let $(\cS^*_\infty[F;Y],L)$ be a graded Lagrangian system.
Describing its Noether identities, we follow the general analysis
of Noether identities of differential operators on fibre bundles
\cite{oper}.

Without a lose of generality, let a Lagrangian $L$ be even. Its
Euler--Lagrange operator $\dl L$ (\ref{0709'}) takes its values
into a graded vector bundle
\mar{41f33}\beq
\ol{VF}=V^*F\op\ot_F\op\w^n T^*X\to F, \label{41f33}
\eeq
where $V^*F$ is the vertical cotangent bundle of $F\to X$. It
however is not a vector bundle over $Y$. Therefore, we restrict
our consideration to the case of a pull-back composite bundle
\mar{41f1}\beq
F=Y\op\times_X F^1\to Y\to X, \label{41f1}
\eeq
where $F^1\to X$ is a vector bundle.

\begin{rem}
Let us introduce the following notation. Given the vertical
tangent bundle $VE$ of a fibre bundle $E\to X$, by its
density-dual bundle is meant a fibre bundle
\mar{41f2}\beq
\ol{VE}=V^*E\op\ot_E \op\w^n T^*X. \label{41f2}
\eeq
If $E\to X$ is a vector bundle, we have
\be
\ol{VE}=\ol E\op\times_X E, \qquad \ol E=E^*\op\ot_X\op\w^n T^*X,
\ee
where $\ol E$ is called the density-dual of $E$. Let
$E=E^0\oplus_X E^1$ be a graded vector bundle over $X$. Its graded
density-dual is defined as $\ol E=\ol E^1\oplus_X \ol E^0$. In
these terms, we treat a composite bundle $F$ as a graded vector
bundle over $Y$ possessing only an odd part. The density-dual
$\ol{VF}$ (\ref{41f2}) of the vertical tangent bundle $VF$ of
$F\to X$ is $\ol{VF}$ (\ref{41f33}). If $F$ is the pull-back
bundle (\ref{41f1}), then
\mar{41f4}\beq
\ol{VF}=((\ol F^1\op\oplus_Y V^*Y)\op\ot_Y\op\w^n T^*X)\op\oplus_Y
F^1 \label{41f4}
\eeq
is a graded vector bundle over $Y$. Given a graded vector bundle
$E=E^0\oplus_Y E^1$ over $Y$, we consider a composite bundle $E\to
E^0\to X$ and a DBGA
\mar{41f5}\beq
\cP^*_\infty[E;Y]=\cS^*_\infty[E;E^0]. \label{41f5}
\eeq
\end{rem}

\begin{lem} \label{41l1} \mar{41l1} One can associate to any
graded Lagrangian system $(\cS^*_\infty[F;Y],L)$ the chain complex
(\ref{v042}) whose one-boundaries vanish on $\Ker \dl L$.
\end{lem}

\begin{proof} Let us consider the density-dual $\ol{VF}$ (\ref{41f4})
of the vertical tangent bundle $VF\to F$, and let us enlarge an
original algebra $\cS^*_\infty[F;Y]$ to the DBGA
$\cP^*_\infty[\ol{VF};Y]$ (\ref{41f5}) with a local generating
basis $(s^A, \ol s_A)$, $[\ol s_A]=([A]+1){\rm mod}\,2$. Following
the physical terminology \cite{barn,gom}, we agree to call its
elements $\ol s_A$ the antifields of antifield number Ant$[\ol
s_A]= 1$. A DBGA $\cP^*_\infty[\ol{VF};Y]$ is endowed with a
nilpotent right graded derivation $\ol\dl=\rdr^A \cE_A$, where
$\cE_A$ are the variational derivatives (\ref{0709'}). Then we
have a chain complex
\mar{v042}\beq
0\lto \im\ol\dl \llr^{\ol\dl} \cP^{0,n}_\infty[\ol{VF};Y]_1
\llr^{\ol\dl} \cP^{0,n}_\infty[\ol{VF};Y]_2 \label{v042}
\eeq
of graded densities of antifield number $\leq 2$. Its
one-boundaries $\ol\dl\Phi$, $\Phi\in
\cP^{0,n}_\infty[\ol{VF};Y]_2$, by very definition, vanish on
$\Ker \dl L$.
\end{proof}

Any one-cycle
\mar{0712}\beq
\Phi= \op\sum_{0\leq|\La|} \Phi^{A,\La}\ol s_{\La A} \om \in
\cP^{0,n}_\infty[\ol{VF};Y]_1\label{0712}
\eeq
of the complex (\ref{v042}) is a differential operator on a fibre
bundle $\ol{VF}$ such that it is linear on fibres of $\ol{VF}\to
F$ and its kernel contains the graded Euler--Lagrange operator
$\dl L$ (\ref{0709'}), i.e.,
\mar{0713}\beq
\ol\dl\Phi=0, \qquad \op\sum_{0\leq|\La|} \Phi^{A,\La}d_\La
\cE_A\om=0. \label{0713}
\eeq
These equalities are Noether identities of an Euler--Lagrange
operator $\dl L$ \cite{jmp05,lmp08,oper}.

In particular, one-chains $\Phi$ (\ref{0712}) are necessarily
Noether identities if they are boundaries. Therefore, these
Noether identities are called trivial. Accordingly, non-trivial
Noether identities modulo the trivial ones are associated to
elements of the first homology $H_1(\ol\dl)$ of the complex
(\ref{v042}). A Lagrangian $L$ is called degenerate if there are
non-trivial Noether identities.

Non-trivial Noether identities can obey first-stage Noether
identities. In order to describe them, let us assume that the
module $H_1(\ol \dl)$ is finitely generated. Namely, there exists
a graded projective $C^\infty(X)$-module $\cC_{(0)}\subset H_1(\ol
\dl)$ of finite rank possessing a local basis $\{\Delta_r\om\}$:
\mar{41f7}\beq
\Delta_r\om=\op\sum_{0\leq|\La|} \Delta_r^{A,\La}\ol s_{\La
A}\om,\qquad \Delta_r^{A,\La}\in \cS^0_\infty[F;Y], \label{41f7}
\eeq
such that any element $\Phi\in H_1(\ol \dl)$ factorizes as
\mar{xx2}\beq
\Phi= \op\sum_{0\leq|\Xi|} \Phi^{r,\Xi} d_\Xi \Delta_r \om, \qquad
\Phi^{r,\Xi}\in \cS^0_\infty[F;Y], \label{xx2}
\eeq
through elements (\ref{41f7}) of $\cC_{(0)}$. Thus, all
non-trivial Noether identities (\ref{0713}) result from Noether
identities
\mar{v64}\beq
\ol\dl\Delta_r= \op\sum_{0\leq|\La|} \Delta_r^{A,\La} d_\La
\cE_A=0, \label{v64}
\eeq
called the complete Noether identities.

\begin{lem} \label{41l2} \mar{41l2} If the homology $H_1(\ol\dl)$
of the complex (\ref{v042}) is finitely generated in the above
mentioned sense, this complex can be extended to the one-exact
chain complex (\ref{v66}) with a boundary operator whose
nilpotency conditions are equivalent to the complete Noether
identities (\ref{v64}).
\end{lem}

\begin{proof}
By virtue of Serre--Swan Theorem \ref{vv0}, a graded module
$\cC_{(0)}$ is isomorphic to a module of sections of the
density-dual $\ol E_0$ of some graded vector bundle $E_0\to X$.
Let us enlarge $\cP^*_\infty[\ol{VF};Y]$ to a DBGA
\mar{41f14}\beq
\ol\cP^*_\infty\{0\}=\cP^*_\infty[\ol{VF}\op\oplus_Y \ol E_0;Y]
\label{41f14}
\eeq
possessing the local generating basis $(s^A,\ol s_A, \ol c_r)$
where $\ol c_r$ are antifields of Grassmann parity $[\ol
c_r]=([\Delta_r]+1){\rm mod}\,2$ and antifield number ${\rm
Ant}[\ol c_r]=2$. The DBGA (\ref{41f14}) is provided with an odd
right graded derivation $\dl_0=\ol\dl + \rdr^r\Delta_r$ which is
nilpotent iff the complete Noether identities (\ref{v64}) hold.
Then $\dl_0$ is a boundary operator of a chain complex
\mar{v66}\beq
0\lto \im\ol\dl \op\lto^{\ol\dl}
\cP^{0,n}_\infty[\ol{VF};Y]_1\op\lto^{\dl_0}
\ol\cP^{0,n}_\infty\{0\}_2 \op\lto^{\dl_0}
\ol\cP^{0,n}_\infty\{0\}_3 \label{v66}
\eeq
of graded densities of antifield number $\leq 3$. Let $H_*(\dl_0)$
denote its homology. We have $H_0(\dl_0)=H_0(\ol\dl)=0$.
Furthermore, any one-cycle $\Phi$ up to a boundary takes the form
(\ref{xx2}) and, therefore, it is a $\dl_0$-boundary
\be
\Phi= \op\sum_{0\leq|\Si|} \Phi^{r,\Xi} d_\Xi \Delta_r\om
=\dl_0(\op\sum_{0\leq|\Si|} \Phi^{r,\Xi}\ol c_{\Xi r}\om).
\ee
Hence, $H_1(\dl_0)=0$, i.e., the complex (\ref{v66}) is one-exact.
\end{proof}

Let us consider the second homology $H_2(\dl_0)$ of the complex
(\ref{v66}). Its two-chains  read
\mar{41f9}\beq
\Phi= G + H= \op\sum_{0\leq|\La|} G^{r,\La}\ol c_{\La r}\om +
\op\sum_{0\leq|\La|,|\Si|} H^{(A,\La)(B,\Si)}\ol s_{\La A}\ol
s_{\Si B}\om. \label{41f9}
\eeq
Its two-cycles define first-stage Noether identities
\mar{v79}\beq
\dl_0 \Phi=0, \qquad   \op\sum_{0\leq|\La|} G^{r,\La}d_\La\Delta_r
\om =-\ol\dl H. \label{v79}
\eeq
Conversely, let the equality (\ref{v79}) hold. Then it is a cycle
condition of the two-chain (\ref{41f9}).

The first-stage Noether identities (\ref{v79}) are trivial either
if a two-cycle $\Phi$ (\ref{41f9}) is a $\dl_0$-boundary or its
summand $G$ vanishes on $\Ker \dl L$. Therefore, non-trivial
first-stage Noether identities fails to exhaust the second
homology $H_2(\dl_0)$ the complex (\ref{v66}) in general.

\begin{lem} \label{v134'} \mar{v134'}
Non-trivial first-stage Noether identities modulo the trivial ones
are identified with elements of the homology $H_2(\dl_0)$ iff any
$\ol\dl$-cycle $\f\in \ol\cP^{0,n}_\infty\{0\}_2$ is a
$\dl_0$-boundary.
\end{lem}

\begin{proof}
It suffices to show that, if the summand $G$ of a two-cycle $\Phi$
(\ref{41f9}) is $\ol\dl$-exact, then $\Phi$ is a boundary. If
$G=\ol\dl \Psi$, let us write
\mar{v169'}\beq
\Phi=\dl_0\Psi +(\ol \dl-\dl_0)\Psi + H. \label{v169'}
\eeq
Hence, the cycle condition (\ref{v79}) reads
\be
\dl_0\Phi=\ol\dl((\ol\dl-\dl_0)\Psi + H)=0.
\ee
Since any $\ol\dl$-cycle $\f\in \ol\cP^{0,n}_\infty\{0\}_2$, by
assumption, is $\dl_0$-exact, then $(\ol \dl-\dl_0)\Psi + H$ is a
$\dl_0$-boundary. Consequently, $\Phi$ (\ref{v169'}) is
$\dl_0$-exact. Conversely, let $\Phi\in
\ol\cP^{0,n}_\infty\{0\}_2$ be a $\ol\dl$-cycle, i.e.,
\be
\ol\dl\Phi= 2\Phi^{(A,\La)(B,\Sigma)}\ol s_{\La A} \ol\dl\ol
s_{\Sigma B}\om= 2\Phi^{(A,\La)(B,\Sigma)}\ol s_{\La A} d_\Si
\cE_B\om=0.
\ee
It follows that $\Phi^{(A,\La)(B,\Sigma)} \ol\dl\ol s_{\Sigma
B}=0$ for all indices $(A,\La)$. Omitting a $\ol\dl$-boundary
term, we obtain
\be
\Phi^{(A,\La)(B,\Sigma)} \ol s_{\Sigma B}= G^{(A,\La)(r,\Xi)}d_\Xi
\Delta_r.
\ee
Hence, $\Phi$ takes the form $\Phi=G'^{(A,\La)(r,\Xi)}
d_\Xi\Delta_r \ol s_{\La A}\om$. Then there exists a three-chain
$\Psi= G'^{(A,\La)(r,\Xi)} \ol c_{\Xi r} \ol s_{\La A}\om$ such
that
\mar{41f12}\beq
\dl_0\Psi=\Phi +\si = \Phi + G''^{(A,\La)(r,\Xi)}d_\La\cE_A \ol
c_{\Xi r} \om. \label{41f12}
\eeq
Owing to the equality $\ol\dl\Phi=0$, we have $\dl_0\si=0$. Thus,
$\si$ in the expression (\ref{41f12}) is $\ol\dl$-exact
$\dl_0$-cycle. By assumption, it is $\dl_0$-exact, i.e.,
$\si=\dl_0\psi$.   Consequently, a $\ol\dl$-cycle $\Phi$ is a
$\dl_0$-boundary $\Phi=\dl_0\Psi -\dl_0\psi$.
\end{proof}

A degenerate Lagrangian system is called reducible if it admits
non-trivial first stage Noether identities.

If the condition of Lemma \ref{v134'} is satisfied, let us assume
that non-trivial first-stage Noether identities are finitely
generated as follows. There exists a graded projective
$C^\infty(X)$-module $\cC_{(1)}\subset H_2(\dl_0)$ of finite rank
possessing a local basis $\{\Delta_{r_1}\om\}$:
\mar{41f13}\beq
\Delta_{r_1}\om=\op\sum_{0\leq|\La|} \Delta_{r_1}^{r,\La}\ol
c_{\La r}\om + h_{r_1}\om,   \label{41f13}
\eeq
such that any element $\Phi\in H_2(\dl_0)$ factorizes as
\mar{v80'}\beq
\Phi= \op\sum_{0\leq|\Xi|} \Phi^{r_1,\Xi} d_\Xi \Delta_{r_1}\om,
\qquad \Phi^{r_1,\Xi}\in \cS^0_\infty[F;Y], \label{v80'}
\eeq
through elements (\ref{41f13}) of $\cC_{(1)}$. Thus, all
non-trivial first-stage Noether identities (\ref{v79}) result from
the equalities
\mar{v82'}\beq
 \op\sum_{0\leq|\La|} \Delta_{r_1}^{r,\La} d_\La \Delta_r +\ol\dl
h_{r_1} =0, \label{v82'}
\eeq
called the complete first-stage Noether identities.

\begin{lem} \label{v139'} \mar{v139'} The one-exact complex
(\ref{v66}) of a reducible Lagrangian system is extended to the
two-exact one (\ref{v87'}) with a boundary operator whose
nilpotency conditions are equivalent to the complete Noether
identities (\ref{v64}) and the complete first-stage Noether
identities (\ref{v82'}).
\end{lem}

\begin{proof}
By virtue of Serre--Swan Theorem \ref{vv0}, a graded module
$\cC_{(1)}$ is isomorphic to a module of sections of the
density-dual $\ol E_1$ of some graded vector bundle $E_1\to X$.
Let us enlarge the DBGA $\ol\cP^*_\infty\{0\}$ (\ref{41f14}) to a
DBGA
\be
\ol\cP^*_\infty\{1\}=\cP^*_\infty[\ol{VF}\op\oplus_Y \ol
E_0\op\oplus_Y\ol E_1;Y]
\ee
possessing a local generating basis $\{s^A,\ol s_A, \ol c_r, \ol
c_{r_1}\}$ where $\ol c_{r_1}$ are first stage Noether antifields
of Grassmann parity $[\ol c_{r_1}]=([\Delta_{r_1}]+1){\rm mod}\,2$
and antifield number Ant$[\ol c_{r_1}]=3$. This DBGA is provided
with an odd right graded derivation $\dl_1=\dl_0 + \rdr^{r_1}
\Delta_{r_1}$ which is nilpotent iff the complete Noether
identities (\ref{v64}) and the complete first-stage Noether
identities (\ref{v82'}) hold. Then $\dl_1$ is a boundary operator
of a chain complex
\mar{v87'}\beq
0\lto \im\ol\dl \op\lto^{\ol\dl}
\cP^{0,n}_\infty[\ol{VF};Y]_1\op\lto^{\dl_0}
\ol\cP^{0,n}_\infty\{0\}_2 \op\lto^{\dl_1}
\ol\cP^{0,n}_\infty\{1\}_3 \op\lto^{\dl_1}
\ol\cP^{0,n}_\infty\{1\}_4 \label{v87'}
\eeq
of graded densities of antifield number $\leq 4$. Let $H_*(\dl_1)$
denote its homology. It is readily observed that
\be
H_0(\dl_1)=H_0(\ol\dl), \qquad H_1(\dl_1)=H_1(\dl_0)=0.
\ee
By virtue of the expression (\ref{v80'}), any two-cycle of the
complex (\ref{v87'}) is a boundary
\be
 \Phi= \op\sum_{0\leq|\Xi|} \Phi^{r_1,\Xi} d_\Xi \Delta_{r_1}\om
=\dl_1(\op\sum_{0\leq|\Xi|} \Phi^{r_1,\Xi} \ol c_{\Xi r_1}\om).
\ee
It follows that $H_2(\dl_1)=0$, i.e., the complex (\ref{v87'}) is
two-exact.
\end{proof}

If the third homology $H_3(\dl_1)$ of the complex (\ref{v87'}) is
not trivial, its elements correspond to second-stage Noether
identities which the complete first-stage ones satisfy, and so on.
Iterating the arguments, one comes to the following.

A degenerate graded Lagrangian system $(\cS^*_\infty[F;Y],L)$ is
called $N$-stage reducible if it admits finitely generated
non-trivial $N$-stage Noether identities, but no non-trivial
$(N+1)$-stage ones. It is characterized as follows
\cite{jmp05a,lmp08}.

$\bullet$ There are graded vector bundles $E_0,\ldots, E_N$ over
$X$, and a DBGA $\cP^*_\infty[\ol{VF};Y]$ is enlarged to a DBGA
\mar{v91}\beq
\ol\cP^*_\infty\{N\}=\cP^*_\infty[\ol{VF}\op\oplus_Y \ol
E_0\op\oplus_Y\cdots\op\oplus_Y \ol E_N;Y] \label{v91}
\eeq
with the local generating basis $(s^A,\ol s_A, \ol c_r, \ol
c_{r_1}, \ldots, \ol c_{r_N})$ where $\ol c_{r_k}$ are Noether
$k$-stage antifields of antifield number Ant$[\ol c_{r_k}]=k+2$.

$\bullet$ The DBGA (\ref{v91}) is provided with the nilpotent
right graded derivation
\mar{v92,'}\ben
&&\dl_{\rm KT}=\dl_N=\ol\dl +
\op\sum_{0\leq|\La|}\rdr^r\Delta_r^{A,\La}\ol s_{\La A} +
\op\sum_{1\leq k\leq N}\rdr^{r_k} \Delta_{r_k},
\label{v92}\\
&& \Delta_{r_k}\om= \op\sum_{0\leq|\La|}
\Delta_{r_k}^{r_{k-1},\La}\ol c_{\La r_{k-1}}\om +
\label{v92'}\\
&& \qquad \op\sum_{0\leq |\Si|, |\Xi|}(h_{r_k}^{(r_{k-2},\Si)(A,\Xi)}\ol
c_{\Si r_{k-2}}\ol s_{\Xi A}+...)\om \in
\ol\cP^{0,n}_\infty\{k-1\}_{k+1}, \nonumber
\een
of antifield number -1. The index $k=-1$ here stands for $\ol
s_A$. The nilpotent derivation $\dl_{\rm KT}$ (\ref{v92}) is
called the Koszul--Tate operator.

$\bullet$ With this graded derivation, the module
$\ol\cP^{0,n}_\infty\{N\}_{\leq N+3}$ of densities of antifield
number $\leq (N+3)$ is decomposed into the exact Koszul--Tate
chain complex
\mar{v94}\ben
&& 0\lto \im \ol\dl \llr^{\ol\dl}
\cP^{0,n}_\infty[\ol{VF};Y]_1\llr^{\dl_0}
\ol\cP^{0,n}_\infty\{0\}_2\llr^{\dl_1}
\ol\cP^{0,n}_\infty\{1\}_3\cdots
\label{v94}\\
&& \qquad
 \llr^{\dl_{N-1}} \ol\cP^{0,n}_\infty\{N-1\}_{N+1}
\llr^{\dl_{\rm KT}} \ol\cP^{0,n}_\infty\{N\}_{N+2}\llr^{\dl_{\rm
KT}} \ol\cP^{0,n}_\infty\{N\}_{N+3} \nonumber
\een
which satisfies the following homology regularity condition.

\begin{note} \label{v155} \mar{v155} Any $\dl_{k<N}$-cycle
$\f\in \ol\cP_\infty^{0,n}\{k\}_{k+3}\subset
\ol\cP_\infty^{0,n}\{k+1\}_{k+3}$ is a $\dl_{k+1}$-boundary.
\end{note}

\begin{rem}
The exactness of the complex (\ref{v94}) means that any
$\dl_{k<N}$-cycle $\f\in \cP_\infty^{0,n}\{k\}_{k+3}$, is a
$\dl_{k+2}$-boundary, but not necessary a $\dl_{k+1}$-one.
\end{rem}

$\bullet$ The nilpotentness $\dl_{\rm KT}^2=0$ of the Koszul--Tate
operator (\ref{v92}) is equivalent to complete non-trivial Noether
identities (\ref{v64}) and complete non-trivial $(k\leq N)$-stage
Noether identities
\mar{v93}\beq
\op\sum_{0\leq|\La|} \Delta_{r_k}^{r_{k-1},\La}d_\La
(\op\sum_{0\leq|\Si|} \Delta_{r_{k-1}}^{r_{k-2},\Si}\ol c_{\Si
r_{k-2}}) = - \ol\dl(\op\sum_{0\leq |\Si|,
|\Xi|}h_{r_k}^{(r_{k-2},\Si)(A,\Xi)}\ol c_{\Si r_{k-2}}\ol s_{\Xi
A}). \label{v93}
\eeq
This item means the following.

\begin{prop} Any $\dl_k$-cocycle $\Phi\in
\cP^{0,n}_\infty\{k\}_{k+2}$ is a $k$-stage Noether identity, and
{\it vice versa}.
\end{prop}

\begin{proof} Any $(k+2)$-chain $\Phi\in \cP^{0,n}_\infty\{k\}_{k+2}$ takes the form
\mar{v156'}\beq
 \Phi= G+H=\op\sum_{0\leq|\La|} G^{r_k,\La}\ol c_{\La r_k}\om +
 \op\sum_{0\leq \Si, 0\leq\Xi}(H^{(A,\Xi)(r_{k-1},\Si)}\ol s_{\Xi
A}\ol c_{\Si r_{k-1}}+...)\om. \label{v156'}
\eeq
If it is a $\dl_k$-cycle, then
\mar{v145'}\beq
 \op\sum_{0\leq|\La|} G^{r_k,\La}d_\La (\op\sum_{0\leq|\Si|}
\Delta_{r_k}^{r_{k-1},\Si}\ol c_{\Si r_{k-1}}) +
\ol\dl(\op\sum_{0\leq \Si, 0\leq\Xi}H^{(A,\Xi)(r_{k-1},\Si)}\ol
s_{\Xi A}\ol c_{\Si r_{k-1}})=0 \label{v145'}
\eeq
are the $k$-stage Noether identities. Conversely, let the
condition (\ref{v145'}) hold. Then it can be extended to a cycle
condition as follows. It is brought into the form
\be
&& \dl_k(\op\sum_{0\leq|\La|} G^{r_k,\La}\ol c_{\La r_k} +
\op\sum_{0\leq \Si, 0\leq\Xi}H^{(A,\Xi)(r_{k-1},\Si)}\ol
s_{\Xi A}\ol c_{\Si r_{k-1}})=\\
&& \qquad  -\op\sum_{0\leq|\La|} G^{r_k,\La}d_\La h_{r_k}
+\op\sum_{0\leq \Si, 0\leq\Xi}H^{(A,\Xi)(r_{k-1},\Si)}\ol s_{\Xi
A}d_\Si \Delta_{r_{k-1}}.
\ee
A glance at the expression (\ref{v92'}) shows that a term in the
right-hand side of this equality belongs to
$\cP^{0,n}_\infty\{k-2\}_{k+1}$. It is a $\dl_{k-2}$-cycle then a
$\dl_{k-1}$-boundary $\dl_{k-1}\Psi$ in accordance with Condition
\ref{v155}. Then the equality (\ref{v145'}) is a $\ol c_{\Si
r_{k-1}}$-dependent part of a cycle condition
\be
\dl_k(\op\sum_{0\leq|\La|} G^{r_k,\La}\ol c_{\La r_k} +
\op\sum_{0\leq \Si, 0\leq\Xi}H^{(A,\Xi)(r_{k-1},\Si)}\ol s_{\Xi
A}\ol c_{\Si r_{k-1}} -\Psi)=0,
\ee
but $\dl_k\Psi$ does not make a contribution to this condition.
\end{proof}

\begin{prop}
Any trivial $k$-stage Noether identity is a $\dl_k$-boundary
$\Phi\in \cP^{0,n}_\infty\{k\}_{k+2}$.
\end{prop}

\begin{proof}
The $k$-stage Noether identities (\ref{v145'}) are trivial either
if a $\dl_k$-cycle $\Phi$ (\ref{v156'}) is a $\dl_k$-boundary or
its summand $G$ vanishes on $\Ker \dl L$. Let us show that, if the
summand $G$ of $\Phi$ (\ref{v156'}) is $\ol\dl$-exact, then $\Phi$
is a $\dl_k$-boundary. If $G=\ol\dl \Psi$, one can write
\be
\Phi=\dl_k\Psi +(\ol \dl-\dl_k)\Psi + H.
\ee
Hence, the $\dl_k$-cycle condition reads
\be
\dl_k\Phi=\dl_{k-1}((\ol\dl-\dl_k)\Psi + H)=0.
\ee
By virtue of Note \ref{v155}, any $\dl_{k-1}$-cycle $\f\in
\ol\cP^{0,n}_\infty\{k-1\}_{k+2}$ is $\dl_k$-exact. Then $(\ol
\dl-\dl_k)\Psi + H$ is a $\dl_k$-boundary. Consequently, $\Phi$
(\ref{v156'}) is $\dl_k$-exact.
\end{proof}

Note that all non-trivial $k$-stage Noether identities
(\ref{v145'}), by assumption, factorize as
\be
\Phi= \op\sum_{0\leq|\Xi|} \Phi^{r_k,\Xi} d_\Xi \Delta_{r_k}\om,
\qquad \Phi^{r_1,\Xi}\in \cS^0_\infty[F;Y],
\ee
through the complete ones (\ref{v93}).

It may happen that a graded Lagrangian system possesses
non-trivial Noether identities of any stage. However, we restrict
our consideration to $N$-reducible Lagrangian systems.

\section{Second Noether theorems}

Different variants of the second Noether theorem have been
suggested in order to relate reducible Noether identities and
gauge symmetries \cite{barn,jmp05,jmp09}. The inverse second
Noether Theorem \ref{w35}, that we formulate in homology terms,
associates to the Koszul--Tate complex (\ref{v94}) of non-trivial
Noether identities the cochain sequence (\ref{w108}) with the
ascent operator $\bu$ (\ref{w108'}) whose components are
non-trivial gauge and higher-stage gauge symmetries.

\begin{rem} \label{42n1} \mar{42n1}
Let us use the following notation. Given the DBGA
$\ol\cP^*_\infty\{N\}$ (\ref{v91}), we consider a DBGA
\mar{w5}\beq
\cP^*_\infty\{N\}=\cP^*_\infty[F\op\oplus_Y E_0\op\oplus_Y\cdots
\op\oplus_Y E_N;Y], \label{w5}
\eeq
possessing a local generating basis $(s^A, c^r, c^{r_1}, \ldots,
c^{r_N})$, $[c^{r_k}]=([\ol c_{r_k}]+1){\rm mod}\,2$, and a DBGA
\mar{w6}\beq
P^*_\infty\{N\}=\cP^*_\infty[\ol{VF}\op\oplus_Y E_0\oplus\cdots
\op\oplus_Y E_N \op\oplus_Y \ol E_0\op\oplus_Y\cdots\op\oplus_Y
\ol E_N;Y] \label{w6}
\eeq
with a local generating basis $(s^A, \ol s_A, c^r, c^{r_1},
\ldots, c^{r_N},\ol c_r, \ol c_{r_1}, \ldots, \ol c_{r_N})$.
Following the physical terminology, we call their elements
$c^{r_k}$ the $k$-stage ghosts of ghost number gh$[c^{r_k}]=k+1$
and antifield number ${\rm Ant}[c^{r_k}]=-(k+1)$. A
$C^\infty(X)$-module $\cC^{(k)}$ of $k$-stage ghosts is the
density-dual of a module $\cC_{(k)}$ of $k$-stage antifields. The
DBGAs $\ol\cP^*_\infty\{N\}$ (\ref{v91}) and $\cP^*_\infty\{N\}$
(\ref{w5}) are subalgebras of $P^*_\infty\{N\}$ (\ref{w6}). The
Koszul--Tate operator $\dl_{\rm KT}$ (\ref{v92}) is naturally
extended to a graded derivation of a DBGA $P^*_\infty\{N\}$.
\end{rem}

\begin{rem} \label{42n10} \mar{42n10} Any
graded differential form $\f\in \cS^*_\infty[F;Y]$ and any finite
tuple $(f^\La)$, $0\leq |\La|\leq k$, of local graded functions
$f^\La\in \cS^0_\infty[F;Y]$ obey the following relations
\cite{book09}:
\mar{qq1}\ben
&& \op\sum_{0\leq |\La|\leq k} f^\La d_\La \f\w \om= \op\sum_{0\leq
|\La|}(-1)^{|\La|}d_\La (f^\La)\f\w \om +d_H\si,
\label{qq1a}\\
&& \op\sum_{0\leq |\La|\leq k} (-1)^{|\La|}d_\La(f^\La \f)=
\op\sum_{0\leq |\La|\leq k} \eta (f)^\La d_\La \f, \nonumber\\
&& \eta (f)^\La = \op\sum_{0\leq|\Si|\leq k-|\La|}(-1)^{|\Si+\La|}
\frac{(|\Si+\La|)!}{|\Si|!|\La|!} d_\Si f^{\Si+\La},
\label{qq1c}\\
&& \eta(\eta(f))^\La=f^\La. \label{qq1d}
\een
\end{rem}

\begin{theo} \label{w35} \mar{w35} Given a Koszul--Tate complex (\ref{v94}),
the module of graded densities $\cP_\infty^{0,n}\{N\}$ is
decomposed into a cochain sequence
\mar{w108,'}\ben
&& 0\to \cS^{0,n}_\infty[F;Y]\ar^{\bu}
\cP^{0,n}_\infty\{N\}^1\ar^{\bu}
\cP^{0,n}_\infty\{N\}^2\ar^{\bu}\cdots, \label{w108}\\
&& \bu=u + u^{(1)}+\cdots + u^{(N)}= u^A\dr_A + u^r\dr_r +\cdots  +
u^{r_{N-1}}\dr_{r_{N-1}}, \label{w108'}
\een
graded in a ghost number. Its ascent operator $\bu$ (\ref{w108'})
is an odd graded derivation of ghost number 1 where $u$
(\ref{w33}) is a variational symmetry of a graded Lagrangian $L$
and the graded derivations $u_{(k)}$ (\ref{w38}), $k=1,\ldots, N$,
obey the relations (\ref{w34}).
\end{theo}

\begin{proof} Given the Koszul--Tate operator (\ref{v92}), let us extend an original
grade Lagrangian $L$ to a Lagrangian
\mar{w8}\beq
L_e=L+L_1=L + \op\sum_{0\leq k\leq N} c^{r_k}\Delta_{r_k}\om=L
+\dl_{\rm KT}( \op\sum_{0\leq k\leq N} c^{r_k}\ol c_{r_k}\om)
\label{w8}
\eeq
of zero antifield number. It is readily observed that a
Koszul--Tate operator $\dl_{\rm KT}$ is an exact symmetry of the
extended Lagrangian $L_e\in P^{0,n}_\infty\{N\}$ (\ref{w8}). Since
a graded derivation $\dl_{\rm KT}$ is vertical, it follows from
the first variational formula (\ref{g107}) that
\mar{w16}\ben
&& \left[\frac{\op\dl^\lto \cL_e}{\dl \ol s_A}\cE_A
+\op\sum_{0\leq k\leq N} \frac{\op\dl^\lto \cL_e}{\dl \ol
c_{r_k}}\Delta_{r_k}\right]\om =  \left[\up^A\cE_A +
\op\sum_{0\leq k\leq N}\up^{r_k}\frac{\dl
\cL_e}{\dl c^{r_k}}\right]\om= d_H\si, \label{w16} \\
&& \up^A= \frac{\op\dl^\lto \cL_e}{\dl \ol s_A}=u^A+w^A
=\op\sum_{0\leq|\La|} c^r_\La\eta(\Delta^A_r)^\La +
 \op\sum_{1\leq i\leq N}\op\sum_{0\leq|\La|}
c^{r_i}_\La\eta(\op\dr^\lto{}^A(h_{r_i}))^\La, \nonumber\\
&& \up^{r_k}=\frac{\op\dl^\lto \cL_e}{\dl \ol c_{r_k}} =u^{r_k}+
w^{r_k}= \op\sum_{0\leq|\La|}
c^{r_{k+1}}_\La\eta(\Delta^{r_k}_{r_{k+1}})^\La +
\op\sum_{k+1<i\leq N} \op\sum_{0\leq|\La|}
c^{r_i}_\La\eta(\op\dr^\lto{}^{r_k}(h_{r_i}))^\La. \nonumber
\een
The equality (\ref{w16}) is split into a set of equalities
\mar{w19,20}\ben
&& \frac{\op\dl^\lto (c^r\Delta_r)}{\dl \ol s_A}\cE_A \om
=u^A\cE_A \om=d_H\si_0, \label{w19}\\
&&  \left[\frac{\op\dl^\lto (c^{r_k}\Delta_{r_k})}{\dl \ol s_A}\cE_A
+\op\sum_{0\leq i<k} \frac{\op\dl^\lto (c^{r_k}\Delta_{r_k})}{\dl
\ol c_{r_i}}\Delta_{r_i}\right] \om= d_H\si_k,  \label{w20}
\een
where $k=1,\ldots,N$. A glance at the equality (\ref{w19}) shows
that, by virtue of the first variational formula (\ref{g107}), an
odd graded derivation
\mar{w33}\beq
u= u^A\frac{\dr}{\dr s^A}, \qquad u^A =\op\sum_{0\leq|\La|}
c^r_\La\eta(\Delta^A_r)^\La, \label{w33}
\eeq
of $\cP^0\{0\}$ is a variational symmetry of a graded Lagrangian
$L$. Every equality (\ref{w20}) falls into a set of equalities
graded by the polynomial degree in antifields. Let us consider
that of them linear in antifields $\ol c_{r_{k-2}}$. We have
\be
&& \frac{\op\dl^\lto}{\dl \ol
s_A}(c^{r_k}\op\sum_{0\leq|\Si|,|\Xi|}h_{r_k}^{(r_{k-2},\Si)(A,\Xi)}
\ol
c_{\Si r_{k-2}}\ol s_{\Xi A})\cE_A\om + \\
&& \qquad \frac{\op\dl^\lto}{\dl \ol
c_{r_{k-1}}}(c^{r_k}\op\sum_{0\leq|\Si|}\Delta_{r_k}^{r'_{k-1},\Si}\ol
c_{\Si r'_{k-1}})\op\sum_{0\leq|\Xi|}
\Delta_{r_{k-1}}^{r_{k-2},\Xi}\ol c_{\Xi r_{k-2}}\om= d_H\si_k.
\ee
This equality is brought into the form
\be
 \op\sum_{0\leq|\Xi|}
(-1)^{|\Xi|}d_\Xi(c^{r_k}\op\sum_{0\leq|\Si|}
h_{r_k}^{(r_{k-2},\Si)(A,\Xi)} \ol c_{\Si r_{k-2}})\cE_A \om +
u^{r_{k-1}}\op\sum_{0\leq|\Xi|} \Delta_{r_{k-1}}^{r_{k-2},\Xi}\ol
c_{\Xi r_{k-2}} \om= d_H\si_k.
\ee
Using the relation (\ref{qq1a}), we obtain an equality
\mar{ddd1}\ben
&& \op\sum_{0\leq|\Xi|} c^{r_k}\op\sum_{0\leq|\Si|}
h_{r_k}^{(r_{k-2},\Si)(A,\Xi)} \ol c_{\Si r_{k-2}} d_\Xi\cE_A\om +
 u^{r_{k-1}}\op\sum_{0\leq|\Xi|}
\Delta_{r_{k-1}}^{r_{k-2},\Xi}\ol c_{\Xi r_{k-2}}\om= d_H\si'_k.
\label{ddd1}
\een
A variational derivative of both its sides with respect to $\ol
c_{r_{k-2}}$ leads to a relation
\mar{w34}\beq
\op\sum_{0\leq|\Si|} d_\Si u^{r_{k-1}}\dr_{r_{k-1}}^\Si
u^{r_{k-2}} =\ol\dl(\al^{r_{k-2}}),\qquad \al^{r_{k-2}} =
-\op\sum_{0\leq|\Si|} \eta(h_{r_k}^{(r_{k-2})(A,\Xi)})^\Si
d_\Si(c^{r_k} \ol s_{\Xi A}), \label{w34}
\eeq
which an odd graded derivation
\mar{w38}\beq
u^{(k)}= u^{r_{k-1}}\dr_{r_{k-1}}=\op\sum_{0\leq|\La|}
c^{r_k}_\La\eta(\Delta^{r_{k-1}}_{r_k})^\La\dr_{r_{k-1}}, \quad
k=1,\ldots,N, \label{w38}
\eeq
satisfies. Graded derivations $u$ (\ref{w33}) and $u^{(k)}$
(\ref{w38}) are assembled into the ascent operator $\bu$
(\ref{w108'}) of the cochain sequence (\ref{w108}).
\end{proof}

A glance at the expression (\ref{w33}) shows that a variational
symmetry $u$ is a linear differential operator on a
$C^\infty(X)$-module $\cC^{(0)}$ of ghosts. Therefore, it is a
gauge symmetry of a graded Lagrangian $L$ which is associated to
the complete Noether identities (\ref{v64}) \cite{jmp09,book09}.
This association is unique due to the following direct second
Noether theorem.

\begin{theo} \mar{825} \label{825}
A variational derivative of the equality (\ref{w19}) with respect
to ghosts $c^r$ leads to the equality
\be
\dl_r(u^A\cE_A
\om)=\op\sum_{0\leq|\La|}(-1)^{|\La|}d_\La(\eta(\Delta^A_r)^\La\cE_A)=
\op\sum_{0\leq|\La|}(-1)^{|\La|} \eta(\eta(\Delta^A_r))^\La
d_\La\cE_A=0,
\ee
which reproduces the complete Noether identities (\ref{v64}) by
means of the relation (\ref{qq1d}).
\end{theo}

Moreover, the gauge symmetry $u$ (\ref{w33}) is complete in the
following sense. Let
\be
\op\sum_{0\leq|\Xi|} C^RG^{r,\Xi}_R d_\Xi \Delta_r \om
\ee
be some projective $C^\infty(X)$-module of finite rank of
non-trivial Noether identities (\ref{xx2}) parameterized by the
corresponding ghosts $C^R$. We have the equalities
\be
&& 0=\op\sum_{0\leq|\Xi|} C^RG^{r,\Xi}_R d_\Xi
(\op\sum_{0\leq|\La|}\Delta_r^{A,\La}d_\La \cE_A)
\om=\op\sum_{0\leq|\La|}(\op\sum_{0\leq|\Xi|}\eta(G^r_R)^\Xi
C^R_\Xi)
\Delta_r^{A,\La}d_\La \cE_A\om+d_H(\si)=\\
&& \qquad
\op\sum_{0\leq|\La|}(-1)^{|\La|}d_\La(\Delta_r^{A,\La}\op\sum_{0\leq|\Xi|}\eta(G^r_R)^\Xi
C^R_\Xi)\cE_A \om +d_H\si =\\
&& \qquad
\op\sum_{0\leq|\La|}\eta(\Delta_r^A)^\La
d_\La(\op\sum_{0\leq|\Xi|}\eta(G^r_R)^\Xi C^R_\Xi)\cE_A \om
+d_H\si=\\
&&\qquad \op\sum_{0\leq|\La|}u_r^{A,\La}d_\La(\op\sum_{0\leq|\Xi|}\eta(G^r_R)^\Xi
C^R_\Xi)\cE_A \om +d_H\si.
\ee
It follows that a graded derivation
\be
d_\La(\op\sum_{0\leq|\Xi|}\eta(G^r_R)^\Xi C^R_\Xi)u_r^{A,\La}\dr_A
\ee
is a variational symmetry of a graded Lagrangian $L$ and,
consequently, its gauge symmetry parameterized by ghosts $C^R$. It
factorizes through the gauge symmetry (\ref{w33}) by putting
ghosts
\be
c^r= \op\sum_{0\leq|\Xi|}\eta(G^r_R)^\Xi C^R_\Xi.
\ee

Turn now to the relation (\ref{w34}). For $k=1$, it takes the form
\be
\op\sum_{0\leq|\Si|} d_\Si u^r\dr_r^\Si u^A =\ol \dl(\al^A)
\ee
of a first-stage gauge symmetry condition on $\Ker \dl L$ which
the non-trivial gauge symmetry $u$ (\ref{w33}) satisfies.
Therefore, one can treat an odd graded derivation
\be
u^{(1)}= u^r\frac{\dr}{\dr c^r}, \qquad u^r=\op\sum_{0\leq|\La|}
c^{r_1}_\La\eta(\Delta^r_{r_1})^\La,
\ee
as a first-stage gauge symmetry associated to the complete
first-stage Noether identities
\be
 \op\sum_{0\leq|\La|} \Delta_{r_1}^{r,\La}d_\La
(\op\sum_{0\leq|\Si|} \Delta_r^{A,\Si}\ol s_{\Si A}) = -
\ol\dl(\op\sum_{0\leq |\Si|, |\Xi|}h_{r_1}^{(B,\Si)(A,\Xi)}\ol
s_{\Si B}\ol s_{\Xi A}).
\ee

Iterating the arguments, one comes to the relation (\ref{w34})
which provides a $k$-stage gauge symmetry condition, associated to
the complete $k$-stage Noether identities (\ref{v93}).

\begin{theo} \mar{826} \label{826}
Conversely, given the $k$-stage gauge symmetry condition
(\ref{w34}), a variational derivative of the equality (\ref{ddd1})
with respect to ghosts $c^{r_k}$ leads to an equality, reproducing
the $k$-stage Noether identities (\ref{v93})
 by means of the relations (\ref{qq1c}) and (\ref{qq1d}).
\end{theo}

This is a higher-stage extension of the direct second Noether
theorem to reducible gauge symmetries. The odd graded derivation
$u_{(k)}$ (\ref{w38}) is called the $k$-stage gauge symmetry. It
is complete as follows. Let
\be
\op\sum_{0\leq|\Xi|} C^{R_k}G^{r_k,\Xi}_{R_k} d_\Xi \Delta_{r_k}
\om
\ee
be a projective $C^\infty(X)$-module of finite rank of non-trivial
$k$-stage Noether identities (\ref{xx2}) factorizing through the
complete ones (\ref{v93}) and parameterized by the corresponding
ghosts $C^{R_k}$. One can show that it defines a $k$-stage gauge
symmetry factorizing through $u^{(k)}$ (\ref{w38}) by putting
$k$-stage ghosts
\be
c^{r_k}= \op\sum_{0\leq|\Xi|}\eta(G^{r_k}_{R_k})^\Xi C^{R_k}_\Xi.
\ee

The odd graded derivation $u_{(k)}$ (\ref{w38}) is said to be the
complete non-trivial $k$-stage gauge symmetry of a Lagrangian $L$.
Thus, components of the ascent operator $\bu$ (\ref{w108'}) are
complete non-trivial gauge and higher-stage gauge symmetries.

\section{Appendix A.}

We quote the following generalization of the abstract de Rham
theorem \cite{hir}. Let
\be
0\to S\ar^h S_0\ar^{h^0} S_1\ar^{h^1}\cdots\ar^{h^{p-1}}
S_p\ar^{h^p} S_{p+1}, \qquad p>1,
\ee
be an exact sequence of sheaves of Abelian groups over a
paracompact topological space $Z$, where the sheaves $S_q$, $0\leq
q<p$, are acyclic, and let
\beq
0\to \G(Z,S)\ar^{h_*} \G(Z,S_0)\ar^{h^0_*}
\G(Z,S_1)\ar^{h^1_*}\cdots\ar^{h^{p-1}_*} \G(Z,S_p)\ar^{h^p_*}
\G(Z,S_{p+1}) \label{+130}
\eeq
be the corresponding cochain complex of sections of these sheaves.

\begin{theo} \label{+132} \mar{+132}
The $q$-cohomology groups of the cochain complex (\ref{+130}) for
$0\leq q\leq p$ are isomorphic to the cohomology groups $H^q(Z,S)$
of $Z$ with coefficients in the sheaf $S$ \cite{jmp,tak2}.
\end{theo}

\section{Appendix B.}

The proof of Theorems \ref{v11} and \ref{v11'} falls into the
following three steps \cite{cmp04,book09,ijgmmp07}.

(I) We start with showing that the complexes (\ref{g111}) --
(\ref{g112}) are locally exact.

\begin{lem} \label{0465} \mar{0465}
If $Y=\Bbb R^{n+k}\to \Bbb R^n$, the complex (\ref{g111}) is
acyclic.
\end{lem}

\begin{proof}
Referring to \cite{barn} for the proof, we summarize a few
formulas. Any horizontal graded form $\f\in \cS^{0,*}_\infty[F;Y]$
admits a decomposition
\mar{0471}\beq
\f=\f_0 + \wt\f, \qquad \wt\f=
\op\int^1_0\frac{d\la}{\la}\op\sum_{0\leq|\La|}s^A_\La
\dr^\La_A\f, \label{0471}
\eeq
where $\f_0$ is an exterior form on $\Bbb R^{n+k}$. Let $\f\in
\cS^{0,m<n}_\infty[F;Y]$ be $d_H$-closed.  Then its component
$\f_0$ (\ref{0471}) is an exact exterior form on $\Bbb R^{n+k}$
and $\wt\f=d_H\xi$, where $\xi$ is given by the following
expressions. Let us introduce an operator
\be
D^{+\nu}\wt\f=\op\int^1_0\frac{d\la}{\la}\sum_{0\leq k}
k\dl^\nu_{(\m_1}\dl^{\al_1}_{\m_2}\cdots\dl^{\al_{k-1}}_{\m_k)}
\la s^A_{(\al_1\ldots\al_{k-1})}
\dr_A^{\m_1\ldots\m_k}\wt\f(x^\m,\la s^A_\La, dx^\m).
\ee
The relation $[D^{+\nu},d_\m]\wt\f=\dl^\nu_\m\wt\f$ holds, and it
leads to a desired expression
\mar{0473}\beq
 \xi=\op\sum_{k=0}\frac{(n-m-1)!}{(n-m+k)!}D^{+\nu} P_k
\dr_\nu\rfloor\wt\f, \qquad P_0=1, \qquad
 P_k=d_{\nu_1}\cdots d_{\nu_k}D^{+\nu_1}\cdots
D^{+\nu_k}. \label{0473}
\eeq
Now let $\f\in \cS^{0,n}_\infty[F;Y]$ be a graded density such
that $\dl\f=0$. Then its component $\f_0$ (\ref{0471}) is an exact
$n$-form on $\Bbb R^{n+k}$ and $\wt\f=d_H\xi$, where $\xi$ is
given by the expression
\mar{0474}\beq
\xi=\op\sum_{|\La|\geq 0}\op\sum_{\Si+\Xi=\La}(-1)^{|\Si|}s^A_\Xi
d_\Si\dr^{\m+\La}_A\wt\f\om_\m. \label{0474}
\eeq
We also quote the homotopy operator (5.107) in \cite{olv} which
leads to the expression
\mar{0477}\ben
&&\xi=\op\int_0^1 I(\f)(x^\m,\la s^A_\La,
dx^\m)\frac{d\la}{\la}, \label{0477}\\
&& I(\f)=\op\sum_{0\leq|\La|}\op\sum_\m
\frac{\La_\m+1}{n-m+|\La|+1} d_\La\left[ \op\sum_{0\leq|\Xi|}
(-1)^\Xi \frac{(\m+\La+\Xi)!}{(\m+\La)!\Xi!}s^A d_\Xi
\dr_A^{\m+\La+\Xi}(\dr_\m\rfloor \f)\right], \nonumber
\een
where $\La!=\La_{\m_1}!\cdots \La_{\m_n}!$, and $\La_\m$ denotes a
number of occurrences of the index $\m$ in $\La$ \cite{olv}. The
graded forms (\ref{0474}) and (\ref{0477}) differ in a $d_H$-exact
graded form.
\end{proof}

\begin{lem} \label{g220} \mar{g220}
If $Y=\Bbb R^{n+k}\to \Bbb R^n$, the complex (\ref{g112}) is
exact.
\end{lem}

\begin{proof}
The fact that a $d_H$-closed graded form $\f\in
\cS^{1,m<n}_\infty[F;Y]$ is $d_H$-exact is derived from Lemma
\ref{0465} as follows. We write
\mar{0445}\beq
\f=\sum\f_A^\La\w \theta^A_\La, \label{0445}
\eeq
where $\f_A^\La\in \cS^{0,m}_\infty[F;Y]$ are horizontal graded
$m$-forms. Let us introduce additional variables $\ol s^A_\La$ of
the same Grassmann parity as $s^A_\La$. Then one can associate to
each graded $(1,m)$-form $\f$ (\ref{0445}) a unique horizontal
graded $m$-form
\mar{0446}\beq
\ol\f=\sum\f_A^\La\ol s^A_\La, \label{0446}
\eeq
whose coefficients are linear in variables $\ol s^A_\La$, and {\it
vice versa}.  Let us put a modified total differential
\be
\ol d_H=d_H + dx^\la\w \op\sum_{0<|\La|}\ol
s^A_{\la+\La}\ol\dr_A^\La,
\ee
acting on graded forms (\ref{0446}), where $\ol\dr^\La_A$ is the
dual of $d\ol s^A_\La$. Comparing the equalities
\be
\ol d_H\ol s^A_\La=dx^\la s^A_{\la+\La},\qquad
d_H\theta^A_\la=dx^\la\w\theta^A_{\la+\La},
\ee
one can  easily justify that $\ol{d_H\f}=\ol d_H\ol\f$. Let the
graded $(1,m)$-form $\f$ (\ref{0445}) be $d_H$-closed. Then the
associated horizontal graded $m$-form $\ol \f$ (\ref{0446}) is
$\ol d_H$-closed and, by virtue of Lemma \ref{0465}, it is $\ol
d_H$-exact, i.e., $\ol \f= \ol d_H \ol\xi$, where $\ol\xi$ is a
horizontal graded $(m-1)$-form given by the expression
(\ref{0473}) depending  on additional variables $\ol s^A_\La$. A
glance at this expression shows that, since $\ol\f$ is linear in
variables $\ol s^A_\La$, so is $\ol\xi=\sum\xi_A^\La\ol s^A_\La$.
It follows that $\f=d_H\xi$ where $\xi=\sum\xi_A^\La\w
\theta^A_\La$. It remains to prove the exactness of the complex
(\ref{g112}) at the last term $\vr(\cS^{1,n}_\infty[F;Y])$. If
\be
\vr(\si)=\op\sum_{0\leq|\La|}(-1)^{|\La|}\theta^A\w
[d_\La(\dr_A^\La\rfloor\si)]= \op\sum_{0\leq|\La|}
(-1)^{|\La|}\theta^A\w [d_\La\si_A^\La]\om=0, \qquad  \si\in
\cS^{1,n}_\infty,
\ee
a direct computation gives
\be
\si=d_H\xi,\qquad \xi=-\op\sum_{0\leq|\La|}\op\sum_{\Si+\Xi=\La}
(-1)^{|\Si|}\theta^A_{\Xi}\w d_\Si\si^{\m+\La}_A \om_\m.
\ee
\end{proof}

(II) Let us now prove Theorems \ref{v11} and \ref{v11'} for a DBGA
$\cQ^*_\infty[F;Y]$. Similarly to $\cS^*_\infty[F;Y]$, the sheaf
$\gQ^*_\infty[F;Y]$ and a DBGA $\cQ^*_\infty[F;Y]$ are decomposed
into Grassmann-graded variational bicomplexes. We consider their
subcomplexes
\mar{v35-8}\ben
&& 0\to \Bbb R\to \gQ^0_\infty[F;Y]\ar^{d_H}\gQ^{0,1}_\infty[F;Y]
\cdots \ar^{d_H} \gQ^{0,n}_\infty[F;Y]\ar^\dl
\vr(\gQ^{1,n}_\infty[F;Y]), \label{v35}\\
&& 0\to \gQ^{1,0}_\infty[F;Y]\ar^{d_H} \gQ^{1,1}_\infty[F;Y]\cdots
\ar^{d_H}\gQ^{1,n}_\infty[F;Y]\ar^\vr \vr(\gQ^{1,n}_\infty[F;Y])\to 0, \label{v36}\\
&& 0\to \Bbb R\to \cQ^0_\infty[F;Y]\ar^{d_H}\cQ^{0,1}_\infty[F;Y]
\cdots \ar^{d_H} \cQ^{0,n}_\infty[F;Y]\ar^\dl
\G(\vr(\gQ^{1,n}_\infty[F;Y])),\label{v37} \\
&&  0\to \cQ^{1,0}_\infty[F;Y]\ar^{d_H}
\cQ^{1,1}_\infty[F;Y]\cdots \ar^{d_H}\cQ^{1,n}_\infty[F;Y]\ar^\vr
\G(\vr(\gQ^{1,n}_\infty[F;Y]))\to 0.\label{v38}
\een
By virtue of Lemmas \ref{0465} and \ref{g220}, the complexes
(\ref{v35}) -- (\ref{v36}) are acyclic. The terms
$\gQ^{*,*}_\infty[F;Y]$ of the complexes (\ref{v35}) --
(\ref{v36}) are sheaves of $\cQ^0_\infty[F;Y]$-modules. Since
$J^\infty Y$ admits the partition of unity just by elements of
$\cQ^0_\infty[F;Y]$, these sheaves are fine and, consequently,
acyclic. By virtue of abstract de Rham Theorem \ref{+132},
cohomology of the complex (\ref{v37}) equals the cohomology of
$J^\infty Y$ with coefficients in the constant sheaf $\Bbb R$ and,
consequently, the de Rham cohomology of $Y$ in accordance with the
isomorphisms (\ref{j19'}). Similarly, the complex (\ref{v38}) is
proved to be exact.

Due to monomorphisms $\cO^*_\infty\to \cS^*_\infty[F;Y]\to
\cQ^*_\infty[F;Y]$ this proof gives something more.

\begin{theo} \label{cmp26'} \mar{cmp26'}
Every $d_H$-closed graded form $\f\in\cQ^{0,m<n}_\infty[F;Y]$
falls into the sum
\mar{g214'}\beq
\f=h_0\si + d_H\xi, \qquad \xi\in \cQ^{0,m-1}_\infty[F;Y],
\label{g214'}
\eeq
where $\si$ is a closed $m$-form on $Y$. Any $\dl$-closed $\f\in
\cQ^{0,n}_\infty[F;Y]$ is the sum
\mar{g215'}\beq
\f=h_0\si + d_H\xi, \qquad \xi\in \cQ^{0,n-1}_\infty[F;Y],
\label{g215'}
\eeq
where $\si$ is a closed $n$-form on $Y$.
\end{theo}

(III) It remains to prove that cohomology of the complexes
(\ref{g111}) -- (\ref{g112}) equals that of the complexes
(\ref{v37}) -- (\ref{v38}).

Let the common symbol $D$ stand for $d_H$ and $\dl$. Bearing in
mind the decompositions (\ref{g214'}) -- (\ref{g215'}), it
suffices to show that, if an element $\f\in \cS^*_\infty[F;Y]$ is
$D$-exact in an algebra $\cQ^*_\infty[F;Y]$, then it is so in an
algebra $\cS^*_\infty[F;Y]$.

Lemma \ref{0465} states that, if $Y$ is a contractible bundle and
a $D$-exact graded form $\f$ on $J^\infty Y$ is of finite jet
order $[\f]$ (i.e., $\f\in \cS^*_\infty[F;Y]$), there exists a
graded form $\varphi\in \cS^*_\infty[F;Y]$ on $J^\infty Y$ such
that $\f=D\varphi$. Moreover, a glance at the expressions
(\ref{0473}) and (\ref{0474}) shows that a jet order $[\varphi]$
of $\varphi$ is bounded by an integer $N([\f])$, depending only on
a jet order of $\f$. Let us call this fact the finite exactness of
an operator $D$.  Lemma \ref{0465} shows that the finite exactness
takes place on $J^\infty Y|_U$ over any domain $U\subset Y$. Let
us prove the following.

\begin{lem} \label{24l10} \mar{24l10}
Given a family $\{U_\al\}$ of disjoint open subsets of $Y$, let us
suppose that the finite exactness takes place on $J^\infty
Y|_{U_\al}$ over every subset $U_\al$ from this family. Then, it
is true on $J^\infty Y$ over the union $\op\cup_\al U_\al$ of
these subsets.
\end{lem}

\begin{proof}
 Let $\f\in\cS^*_\infty[F;Y]$ be a $D$-exact graded form on
$J^\infty Y$. The finite exactness on $(\pi^\infty_0)^{-1}(\cup
U_\al)$ holds since $\f=D\varphi_\al$ on every
$(\pi^\infty_0)^{-1}(U_\al)$ and $[\varphi_\al]< N([\f])$.
\end{proof}

\begin{lem} \label{24l11} \mar{24l11}
Suppose that the finite exactness of an operator $D$ takes place
on $J^\infty Y$ over open subsets $U$, $V$ of $Y$ and their
non-empty overlap $U\cap V$. Then, it also is true on $J^\infty
Y|_{U\cup V}$.
\end{lem}

\begin{proof} Let $\f=D\varphi\in\cS^*_\infty[F;Y]$ be a
$D$-exact form on $J^\infty Y$. By assumption, it can be brought
into the form $D\varphi_U$ on $(\pi^\infty_0)^{-1}(U)$ and
$D\varphi_V$ on $(\pi^\infty_0)^{-1}(V)$, where $\varphi_U$ and
$\varphi_V$ are graded forms of bounded jet order. Let us consider
their difference $\varphi_U-\varphi_V$ on
$(\pi^\infty_0)^{-1}(U\cap V)$. It is a $D$-exact graded form of
bounded jet order $[\varphi_U-\varphi_V]< N([\f])$ which, by
assumption, can be written as $\varphi_U-\varphi_V=D\si$ where
$\si$ also is of bounded jet order $[\si]<N(N([\f]))$. Lemma
\ref{am20} below shows that $\si=\si_U +\si_V$ where $\si_U$ and
$\si_V$ are graded forms of bounded jet order on
$(\pi^\infty_0)^{-1}(U)$ and $(\pi^\infty_0)^{-1}(V)$,
respectively. Then, putting
\be
\varphi'|_U=\varphi_U-D\si_U, \qquad \varphi'|_V=\varphi_V+
D\si_V,
\ee
we have a graded form $\f$, equal to $D\varphi'_U$ on
$(\pi^\infty_0)^{-1}(U)$ and $D\varphi'_V$ on
$(\pi^\infty_0)^{-1}(V)$, respectively. Since the difference
$\varphi'_U -\varphi'_V$ on $(\pi^\infty_0)^{-1}(U\cap V)$
vanishes, we obtain $\f=D\varphi'$ on $(\pi^\infty_0)^{-1}(U\cup
V)$ where
\be
\varphi'=\left\{
\begin{array}{ll}
\varphi'|_U=\varphi'_U, &\\
\varphi'|_V=\varphi'_V &
\end{array}\right.
\ee
is of bounded jet order $[\varphi']<N(N([\f]))$.
\end{proof}

\begin{lem} \label{am20} \mar{am20}
Let $U$ and $V$ be open subsets of a bundle $Y$ and $\si\in
\gO^*_\infty$ a graded form of bounded jet order on
$(\pi^\infty_0)^{-1}(U\cap V)\subset J^\infty Y$. Then, $\si$ is
decomposed into  a sum $\si_U+ \si_V$ of graded forms $\si_U$ and
$\si_V$ of bounded jet order on $(\pi^\infty_0)^{-1}(U)$ and
$(\pi^\infty_0)^{-1}(V)$, respectively.
\end{lem}

\begin{proof}
By taking a smooth partition of unity on $U\cup V$ subordinate to
a cover $\{U,V\}$ and passing to a function with support in $V$,
one gets a smooth real function $f$ on $U\cup V$ which equals 0 on
a neighborhood of $U\setminus V$ and 1 on a neighborhood of
$V\setminus U$ in $U\cup V$. Let $(\pi^\infty_0)^*f$ be the
pull-back of $f$ onto $(\pi^\infty_0)^{-1}(U\cup V)$. A graded
form $((\pi^\infty_0)^*f)\si$ equals 0 on a neighborhood of
$(\pi^\infty_0)^{-1}(U)$ and, therefore, can be extended by 0 to
$(\pi^\infty_0)^{-1}(U)$. Let us denote it $\si_U$. Accordingly, a
graded form $(1-(\pi^\infty_0)^*f)\si$ has an extension $\si_V$ by
0 to $(\pi^\infty_0)^{-1}(V)$. Then, $\si=\si_U +\si_V$ is a
desired decomposition because $\si_U$ and $\si_V$ are of the jet
order which does not exceed that of $\si$.
\end{proof}

To prove the finite exactness of $D$ on $J^\infty Y$, it remains
to choose an appropriate cover of $Y$. A smooth manifold $Y$
admits a countable cover $\{U_\xi\}$ by domains $U_\xi$, $\xi\in
{\bf N}$, and its refinement $\{U_{ij}\}$, where $j\in {\bf N}$
and $i$ runs through a finite set, such that $U_{ij}\cap
U_{ik}=\emptyset$, $j\neq k$ \cite{gre}. Then $Y$ has a finite
cover $\{U_i=\cup_j U_{ij}\}$. Since the finite exactness of an
operator $D$ takes place over any domain $U_\xi$, it also holds
over any member $U_{ij}$ of the refinement $\{U_{ij}\}$ of
$\{U_\xi\}$ and, in accordance with Lemma \ref{24l10}, over any
member of a finite cover $\{U_i\}$ of $Y$. Then by virtue of Lemma
\ref{24l11}, the finite exactness of $D$ takes place on $J^\infty
Y$ over $Y$.

Similarly, one can show that, restricted to
$\cS^{k,n}_\infty[F;Y]$, the operator $\vr$ remains exact.

\end{document}